\documentclass{llncs}

\usepackage{amssymb,boxedminipage,amsmath,tikz,subcaption}
\usepackage[export]{adjustbox}

\spnewtheorem{observation}{Observation}{\normalfont\bfseries}{\itshape}
\spnewtheorem{Claim}{Claim}{\normalfont\bfseries}{\itshape}

\newcommand{\Oh}{\mathcal{O}}

\newcommand{\app}{$\spadesuit$}

\usepackage{tikz}
\usetikzlibrary{arrows}
\usetikzlibrary{arrows.meta}
\usetikzlibrary{decorations.pathmorphing}
\usetikzlibrary{decorations.markings}
\usetikzlibrary{calc}
\tikzset{
  circ/.style = {circle,draw,fill,inner sep=1.3pt},
  circr/.style = {circle,draw=red,fill=red,inner sep=1.3pt},
  scirc/.style = {circle,draw,fill,inner sep=.5pt},
  invisible/.style = {draw=none,inner sep=0pt,font=\tiny},
  nonedge/.style={decorate,decoration={snake,amplitude=.3mm,segment length=1mm},draw}
}

\newcommand\yes{\textsc{Yes}}
\newcommand\no{\textsc{No}}

\newcommand\dom{\textsc{Dominating Set}}
\newcommand\contracd{\textsc{1-Edge Contraction($\gamma$)}}

\newcommand\edgecontracd{\textsc{Edge Contraction($\gamma$)}}

\newcommand\kcontracd{\textsc{$k$-Edge Contraction($\gamma$)}}
\newcommand\twocontracd{\textsc{2-Edge Contraction($\gamma$)}}

\title{Reducing the domination number of graphs via edge contractions}
\author{Esther Galby\inst{1}
\and Paloma T. Lima\inst{2} \and Bernard Ries\inst{1}}

\institute{University of Fribourg, Fribourg, Switzerland,  \texttt{esther.galby@unifr.ch,bernard.ries@unifr.ch}
\and
University of Bergen, Bergen, Norway, \texttt{paloma.lima@uib.no}
}

\begin{document}
\maketitle
\setcounter{footnote}{0}


\begin{abstract}
In this paper, we study the following problem: given a connected graph $G$, can we reduce the domination number of $G$ by at least one using $k$ edge contractions, for some fixed integer $k\geq 0$? We present positive and negative results regarding the computational complexity of this problem.
\end{abstract}


\section{Introduction}
\label{s-intro}

In a graph modification problem, we are usually interested in modifying a given graph $G$, via a small number of operations, into some other graph $G'$ that has a certain desired property. This property often describes a certain graph class to which $G'$ must belong. Such graph modification problems allow to capture a variety of classical graph-theoretic problems. Indeed, if for instance only $k$ vertex deletions are allowed and $G'$ must be a stable set or a clique, we obtain the {\sc Stable Set} or {\sc Clique} problem, respectively. 

Now, instead of specifying a graph class to which $G'$ should belong, we may ask for a specific graph parameter $\pi$ to decrease. In other words, given a graph $G$, a set $\mathcal{O}$ of one or more graph operations and an integer~$k\geq 1$, the question is whether $G$ can be transformed into a graph $G'$ by using at most $k$ operations from $\mathcal{O}$ such that $\pi(G')\leq \pi(G)-d$ for some {\it threshold} $d\geq 0$. Such problems are called {\it blocker problems} as the set of vertices or edges involved can be viewed as ``blocking'' the parameter $\pi$. Notice that identifying such sets may provide important information relative to the structure of the graph~$G$.

Blocker problems have been well studied in the literature (see for instance \cite{BTT11,bazgan2013critical,Bentz,CWP11,DPPR15,diner2018contraction,PBP,PPR16,paulusma2017blocking,paulusma2018critical,RBPDCZ10}) and relations to other well-known graph problems have been presented (see for instance \cite{DPPR15,PPR16}). So far, the literature mainly focused on the following graph parameters:  the chromatic number, the independence number, the clique number, the matching number and the vertex cover number. Furthermore, the set $\mathcal{O}$ consisted of a single graph operation, namely either vertex deletion, edge contraction, edge deletion or edge addition. Since these blocker problems are usually $\mathsf{NP}$-hard in general graphs, a particular attention has been paid to their computational complexity when restricted to special graph classes.

In this paper, we focus on another parameter, namely the domination number~$\gamma$, and we restrict $\mathcal{O}$ to a single graph operation, the edge contraction. More specifically, let $G=(V,E)$ be a graph. The {\it contraction} of an edge $uv\in E$ removes vertices $u$ and $v$ from $G$ and replaces them by a new vertex that is made adjacent to precisely those vertices that were adjacent to $u$ or $v$ in $G$ (without introducing self-loops nor multiple edges). We say that a graph $G$ can be \emph{$k$-contracted} into a graph~$G'$, if $G$ can be transformed into $G'$ by a sequence of at most~$k$ edge contractions, for an integer $k\geq 1$. We will be interested in the following problem, where $k\geq 1$ is a fixed integer.

\begin{center}
\begin{boxedminipage}{.99\textwidth}
$k$-\textsc{Edge Contraction($\gamma$)}\\
\begin{tabular}{ r p{0.8\textwidth}}
\textit{~~~~Instance:} &A connected graph $G=(V,E)$\\
\textit{Question:} &Can $G$ be $k$-edge contracted into a graph $G'$ such that $\gamma(G')~\leq~\gamma(G) -1$?
\end{tabular}
\end{boxedminipage}
\end{center}

In other words, we are interested in a blocker problem with parameter $\gamma$, graph operations set $\mathcal{O}=\{$edge contraction$\}$ and threshold $d=1$. Notice that if $\gamma(G)=1$ that is, $G$ contains a dominating vertex, then $G$ is always a \no-instance for $k$-\textsc{Edge Contraction($\gamma$)}. Reducing the domination number using edge contractions was first considered in \cite{HX10}; given a graph $G=(V,E)$, the authors denote by $ct_{\gamma}(G)$ the minimum number of edge contractions required to transform $G$ into a graph $G'$ such that $\gamma (G') \leq \gamma (G) -1$ and prove that for a connected graph $G$ such that $\gamma(G)\geq 2$, we have $ct_\gamma (G)\leq 3$. It follows that a graph $G$ with $\gamma(G)\geq 2$ is always a \yes-instance of $k$-\textsc{Edge Contraction($\gamma$)}, if $k\geq 3$. The authors \cite{HX10} further give necessary and sufficient conditions for $ct_\gamma (G)$ to be equal to 1, respectively~2.

\begin{theorem}[\cite{HX10}]
\label{theorem:contracdom}
For a connected graph $G$, the following holds.
\begin{itemize}
\item[(i)] $ct_\gamma (G)=1$ if and only if there exists a minimum dominating set in $G$ that is not a stable set.
\item[(ii)] $ct_\gamma (G)=2$ if and only if every minimum dominating set in $G$ is a stable set and there exists a dominating set $D$ in $G$ of size $\gamma(G)+1$ such that $G[D]$ contains at least two edges.
\end{itemize}
\end{theorem}

To the best of our knowledge, a systematic study of the complexity of $k$-\textsc{Edge Contraction($\gamma$)} has not yet been attempted in the literature. We here initiate such a study as it has been done for other parameters and several graph operations. Our paper is organised as follows\footnote{Missing proofs will be marked by \app\ and are in the appendix for reviewing purposes.}. In Section \ref{s-pre}, we present definitions and notations that are used throughout the paper. In Section \ref{s-hard}, we prove the ($\mathsf{co}$)$\mathsf{NP}$-hardness of $k$-\textsc{Edge Contraction($\gamma$)} for $k=1,2$. We further show that $1$-\textsc{Edge Contraction($\gamma$)} is $\mathsf{W}$[1]-hard parameterized by the size of a minimum dominating set plus the mim-width of the input graph, and that it remains $\mathsf{NP}$-hard when restricted to $P_9$-free graphs, bipartite graphs and $\{C_3,\ldots , C_l\}$-free graphs for any $l \geq 3$. Finally, we present in Section \ref{s-poly} some positive results; in particular, we show that for any $k\geq 1$, $k$-\textsc{Edge Contraction($\gamma$)} is polynomial-time solvable for $P_5$-free graphs and that it can be solved in $\mathsf{FPT}$-time and $\mathsf{XP}$-time when parameterized by tree-width and mim-width, respectively.


\section{Preliminaries}
\label{s-pre}

Throughout the paper, we only consider finite, undirected, connected graphs that have no self-loops nor multiple edges. We refer the reader to~\cite{Di05} for any terminology and notation not defined here and to \cite{CYGAN} for basic definitions and terminology regarding parameterized complexity.

Let $G=(V,E)$ be a graph and let $u\in V$. We denote by $N_G(u)$, or simply $N(u)$ if it is clear from the context, the set of vertices that are adjacent to $u$ i.e., the {\it neighbors} of $u$, and let $N[u]=N(u)\cup \{u\}$. Two vertices $u,v\in V$ are said to be {\it true twins} (resp. \textit{false twins}), if $N[u]=N[v]$ (resp. if $N(u)=N(v)$).

For a family $\{H_1,\ldots,H_p\}$ of graphs, $G$ is said to be {\it $\{H_1,\ldots,H_p\}$-free} if $G$ has no induced subgraph isomorphic to a graph in $\{H_1,\ldots,H_p\}$; if $p=1$ we may write $H_1$-free instead of $\{H_1\}$-free.  
For a subset $V'\subseteq V$, we let $G[V']$ denote the subgraph of $G$ {\it induced} by $V'$, which has vertex set~$V'$ and edge set $\{uv\in E\; |\; u,v\in V'\}$.

We denote by $d_G(u,v)$, or simply $d(u,v)$ if it is clear from the context, the length of a shortest path from $u$ to $v$ in $G$. Similarly, for any subset $V'\subseteq V$, we denote by $d_G(u,V')$, or simply $d(u,V')$ if it is clear from the context, the minimum length of a shortest path from $u$ to some vertex in $V'$ i.e., $d(u,V')=\min_{v\in V'}d(u,v)$.

For a vertex $v\in V$, we write $G-v=G[V\setminus \{v\}]$ and for a subset $V'\subseteq V$ we write $G-V'=G[V\setminus V']$. For an edge $e \in E$, we denote by $G\backslash e$ the graph obtained from $G$ by contracting the edge $e$. The $k$-{\it subdivision} of an edge $uv$ consists in replacing it by a path $u$-$v_1$-$\ldots$-$v_k$-$v$, where $v_1,\ldots,v_k$ are new vertices.

For $n\geq 1$, the path and cycle on $n$ vertices are denoted by $P_n$ and $C_n$ respectively. A graph is {\it bipartite} if every cycle contains an even number of vertices.
 
A subset $S\subseteq V$ is called an {\it stable set} of $G$ if any two vertices in $S$ are nonadjacent; we may also say that $S$ is \textit{stable}. A subset $D\subseteq V$ is called a {\it dominating set}, if every vertex in $V\setminus D$ is adjacent to at least one vertex in $D$; the {\it domination number} $\gamma(G)$ is the number of vertices in a minimum dominating set. For any $v \in D$ and $u \in N[v]$, $v$ is said to \textit{dominate} $u$ (in particular, $v$ dominates itself); furthermore, $u$ is a \textit{private neighbor of $v$ with respect to $D$} if $u$ has no neighbor in $D\backslash \{v\}$. We say that \textit{$D$ contains an edge} (or more) if the graph $G[D]$ contains an edge (or more). The {\sc Dominating Set} problem is to test whether a given graph $G$ has a dominating set of size at most~$\ell$, for some given integer $\ell\geq 0$. 


\section{Hardness results}
\label{s-hard}

In this section, we present hardness results for the $k$-\textsc{Edge Contraction($\gamma$)} problem. Recall that for $k\geq 3$, the problem is trivial; we show that for $k=1,2$, it becomes ($\mathsf{co}$)$\mathsf{NP}$-hard. To this end, we introduce the following problem.

\begin{center}
\begin{boxedminipage}{.99\textwidth}
\textsc{\sc Contraction Number($\gamma$,$k$)}\\[2pt]
\begin{tabular}{ r p{0.8\textwidth}}
\textit{~~~~Instance:} &A connected graph $G=(V,E)$.\\
\textit{Question:} &Is $ct_{\gamma}(G) = k$?
\end{tabular}
\end{boxedminipage}
\end{center}

\begin{theorem}
\label{thm:ct3}
{\sc Contraction Number($\gamma$,$3$)} is $\mathsf{NP}$-hard.
\end{theorem}

\begin{proof}
We reduce from \textsc{1-in-3 Positive 3-Sat}, where each variable occurs only positively, each clause contains exactly three positive literals, 
and we want a truth assignment such that each clause contains exactly one true variable. This problem is known to be $\mathsf{NP}$-complete \cite{garey}. Given an instance $\Phi$ of this problem, with variable set $X$ and clause set $C$, we construct an equivalent instance $G_{\Phi}$ of {\sc Contraction Number($\gamma$,$3$)} as follows. For any variable $x \in X$, we introduce a copy of $C_3$, which we denote by $G_x$, with two distinguished \textit{truth vertices} $T_x$ and $F_x$ (see Fig. \ref{fig:clausegad}); in the following, the third vertex of $G_x$ is denoted by $u_x$. For any clause $c \in C$ containing variables $x_1,x_2$ and $x_3$, we introduce the gadget $G_c$ depicted in Fig. \ref{fig:clausegad} (where it is connected to the corresponding variable gadgets). The vertex set of the clique $K_c$ corresponds to the set of subsets of size 1 of $\{x_1,x_2,x_3\}$ (hence the notation); for any $i \in \{1,2,3\}$, the vertex $x_i$ (resp. $x'_i$) is connected to every vertex $v_S\in K_c$ such that $x_i \not\in S$ (resp. $x_i \in S$). Finally, for $i=1,2,3$, we add an edge between $t_i$ (resp. $x'_i$) and the truth vertex $T_{x_i}$ (resp. $F_{x_i}$). Our goal now is to show that $\Phi$ is satisfiable if and only if $ct_{\gamma} (G_{\Phi}) = 3$. In the remainder of the proof, given a clause $c \in C$, we denote by $x_1$, $x_2$ and $x_3$ the variables occuring in $c$ and thus assume that $t_i$ (resp. $x'_i$) is adjacent to $T_{x_i}$ (resp. $F_{x_i}$) for $i \in \{1,2,3\}$. Let us first start with some easy observations.

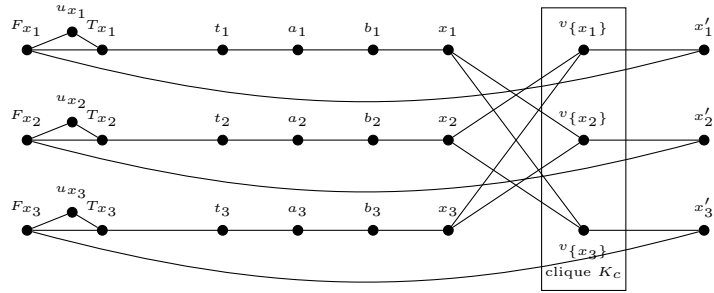
\begin{figure}[htb]
\centering
\begin{tikzpicture}[node distance=1cm,scale=.8]
\node[circ,label=above:{\tiny $T_{x_3}$}] (T3) at (-2,0) {};
\node[circ,label=above:{\tiny $u_{x_3}$}] (u3) at (-2.5,.3) {};
\node[circ,label=above:{\tiny $F_{x_3}$},left of=T3] (F3) {};

\node[circ,label=above:{\tiny $T_{x_2}$}] (T2) at (-2,1.5) {};
\node[circ,label=above:{\tiny $u_{x_2}$}] (u2) at (-2.5,1.8) {};
\node[circ,label=above:{\tiny $F_{x_2}$},left of=T2] (F2) {};

\node[circ,label=above:{\tiny $T_{x_1}$}] (T1) at (-2,3) {};
\node[circ,label=above:{\tiny $u_{x_1}$}] (u1) at (-2.5,3.3) {};
\node[circ,label=above:{\tiny $F_{x_1}$},left of=T1] (F1) {};

\node[circ,label=above:{\tiny $t_3$}] (t3) at (0,0) {};
\node[circ,label=above:{\tiny $a_3$},right of =t3] (a3) {};
\node[circ,label=above:{\tiny $b_3$},right of =a3] (b3) {};
\node[circ,label=above:{\tiny $x_3$},right of =b3] (l3) {};

\node[circ,label=above:{\tiny $t_2$}] (t2) at (0,1.5) {};
\node[circ,label=above:{\tiny $a_2$},right of =t2] (a2) {};
\node[circ,label=above:{\tiny $b_2$},right of =a2] (b2) {};
\node[circ,label=above:{\tiny $x_2$},right of =b2] (l2) {};

\node[circ,label=above:{\tiny $t_1$}] (t1) at (0,3) {};
\node[circ,label=above:{\tiny $a_1$},right of =t1] (a1) {};
\node[circ,label=above:{\tiny $b_1$},right of =a1] (b1) {};
\node[circ,label=above:{\tiny $x_1$},right of =b1] (l1) {};

\node[circ,label=below:{\tiny $v_{\{x_3\}}$}] (x3) at (6,0) {};
\node[circ,label=above:{\tiny $v_{\{x_2\}}$}] (x2) at (6,1.5) {};
\node[circ,label=above:{\tiny $v_{\{x_1\}}$}] (x1) at (6,3) {};

\node[circ,label=above:{\tiny $x'_3$}] (l'3) at (8,0) {};
\node[circ,label=above:{\tiny $x'_2$}] (l'2) at (8,1.5) {};
\node[circ,label=above:{\tiny $x'_1$}] (l'1) at (8,3) {};

\draw (F3) -- (l3)
(F2) -- (l2)
(F1) -- (l1);

\draw (T3) -- (u3)
(F3) -- (u3)
(T2) -- (u2)
(F2) -- (u2)
(T1) -- (u1)
(F1) -- (u1);

\draw (l1) -- (x2)
(l1) -- (x3);

\draw (l'1) -- (x1);

\draw (l2) -- (x1)
(l2) -- (x3);

\draw (l'2) -- (x2);

\draw (l3) -- (x1)
(l3) -- (x2);

\draw (l'3) -- (x3);

\draw (F1) edge[bend right=15] (l'1);
\draw (F2) edge[bend right=15] (l'2);
\draw (F3) edge[bend right=15] (l'3);

\draw (5.3,-1) rectangle (6.7,3.7);
\node[draw=none] at (6,-.7) {\tiny clique $K_c$};
\end{tikzpicture}
\caption{The gadget $G_c$ for a clause $c \in C$ containing variables $x_1$, $x_2$ and $x_3$ (the rectangle indicates that the corresponding set of vertices induces a clique).}
\label{fig:clausegad}
\end{figure}

\begin{observation}
\label{obs:S}
Let $D$ be a dominating set of $G_{\Phi}$. Then for any $x \in X$, $\vert D \cap V(G_x) \vert \geq 1$ and for any $c \in C$, $\vert D \cap V(G_c) \vert \geq 4$. In particular, $\vert D \vert \geq \vert X \vert + 4 \vert C \vert$.
\end{observation}

Clearly, for any $x \in X$, $\vert D \cap G_x \vert \geq 1$ since $u_x$ must be dominated. Also, in order to dominate vertices $a_1,a_2,a_3$ and $v_{\{x_1\}}$ in some gadget $G_c$, we need at least 4 distinct vertices, since their neighborhoods are pairwise disjoint and so, $\vert G \cap V(G_c) \vert \geq 4$, for any $c \in C$.

\begin{observation}
\label{obs:Si}
Let $D$ be a dominating set of $G_{\Phi}$. For any clause gadget $G_c$ and $i \in \{1,2,3\}$, $D \cap \{a_i,b_i,x_i\} \neq \emptyset$.
\end{observation}

This immediately follows from the fact that every vertex $b_i$ needs to be dominated and its neighbors are $a_i$ and $x_i$ for $i\in \{1,2,3\}$.

\begin{observation}
\label{obs:Snotdom}
Let $D$ be a dominating set of $G_{\Phi}$. For any clause gadget $G_c$, if $\vert D \cap V(G_c) \vert = 4$, then $D \cap \{t_i,x'_i\} = \emptyset$ and $\vert D \cap \{a_i,b_i,x_i\} \vert =1$, for any $i \in \{1,2,3\}$. 
\end{observation}

If $t_i \in D$ for some $i \in \{1,2,3\}$, then it follows from Observation \ref{obs:Si} that $\vert D \cap \{a_j,b_j,x_j\} \vert = 1$ for any $j \in \{1,2,3\}$. This implies that at least two vertices among $x_1,x_2$ and $x_3$ belong to $D$ for otherwise there would exist $j\in \{1,2,3\}$ such that $v_{\{x_j\}}$ is not dominated. In particular, there must exist $j \neq i$ such that $x_j \in D$; but then, $a_j$ is not dominated. Similarly, if $x'_i \in D$ for some $i \in \{1,2,3\}$, it follows from Observation \ref{obs:Si} that  $\vert D \cap \{a_j,b_j,x_j\} \vert = 1$ for any $j \in \{1,2,3\}$. But then, in order to dominate the vertices of $K_c$, either $x_i \in D$ but then $a_i$ is not dominated; or $\{x_j, j\neq i\} \subset D$ and $a_j$ with $j \neq i$, is not dominated.

Now suppose that $\vert D \cap \{a_i,b_i,x_i\} \vert \geq 2$ for some $i\in \{1,2,3\}$. Then by Observation~\ref{obs:Si}, we conclude that $\vert D \cap \{a_k,b_k,x_k\} \vert = 1$ for $k \neq i$ and $\vert D \cap \{a_i,b_i,x_i\} \vert = 2$. This implies that $D \cap V(K_c) = \emptyset$ for otherwise we would have $\vert D \cap V(G_c) \vert \geq 5$. But then, since $x'_i \not\in D$, $D$ must contain at least two vertices among $x_1,x_2$ and $x_3$ in order to dominate the vertices of $K_c$; in particular, there exists $j \neq i$ such that $x_j \in D$ and so, $a_j$ is not dominated. 

\begin{observation}
\label{obs:ct3}
Let $D$ be a minimum dominating set of $G_{\Phi}$ and suppose that $ct_{\gamma}(G_{\Phi}) = 3$. Then for any vertices $u,v\in D$, we have $d(u,v)\geq 3$. 
\end{observation}

Indeed, if $u,v$ are at distance at most 2, we conclude by Theorem \ref{theorem:contracdom}(ii) that $ct_{\gamma}(G_{\Phi}) \leq 2$, a contradiction.

\begin{observation}
\label{obs:aifx}
Let $D$ be a minimum dominating set of $G_{\Phi}$ and suppose that $ct_{\gamma}(G_{\Phi}) = 3$. Then for any clause gadget $G_c$ and $i \in \{1,2,3\}$, $a_i \in D$ if and only if $T_{x_i} \not\in D$.
\end{observation}

This readily follows from Observation~\ref{obs:ct3}. Further note that we may assume that for any $i \in \{1,2,3\}$, $a_i \in D$ if and only if $F_{x_i} \in D$; $T_{x_i} \not\in D$ is equivalent to $\{F_{x_i},u_{x_i}\} \cap D \neq \emptyset$ and we may always replace $D$ by $(D \backslash \{u_{x_i}\}) \cup \{F_{x_i}\}$.

\begin{observation}
\label{obs:ai}
Let $D$ be a minimum dominating set of $G_{\Phi}$ and suppose that $ct_{\gamma}(G_{\Phi}) = 3$. Then for any clause gadget $G_c$, $\vert D\cap \{a_1,a_2,a_3\}\vert \leq 2$.
\end{observation}

If it weren't the case then, by Observation \ref{obs:ct3}, no $x_i$ or $b_i$ ($i=1,2,3$) would belong to $D$. But since $x_1,x_2$ and $x_3$ must be dominated, it follows that $D \cap V(K_c) \neq \emptyset$ and by Observation \ref{obs:aifx}, we conclude that $D$ contains two vertices at distance two (namely, $v_{\{x_i\}} \in D \cap V(K_c)$ and $F_{x_i}$ for some $i\in \{1,2,3\}$).

\begin{observation}
\label{obs:bi}
Let $D$ be a minimum dominating set of $G_{\Phi}$ and suppose that $ct_{\gamma}(G_{\Phi}) = 3$. Then for any clause gadget $G_c$, $\vert D\cap \{b_1,b_2,b_3\}\vert \leq 1$.
\end{observation}

Indeed, if we assume, without loss of generality, that $b_1,b_2 \in D$, then by Observation \ref{obs:ct3}, $D \cap V(K_c) = \emptyset$. Furthermore, if $a_k \not\in D$ then $b_k \in D$ since $a_k$ must be dominated and $t_k \not\in D$ (it would otherwise be within distance at most 2 from a vertex in $D$ belonging to $G_{x_k}$); and, if $a_k \in D$, then $b_k \not\in D$ by Observation \ref{obs:ct3}. Thus, we conclude by Observation \ref{obs:aifx} that $b_k \in D$ if and only if $T_{x_k} \in D$. It then follows from Observation \ref{obs:ct3} that $x'_1,x'_2 \not\in D$ and so, $x_3,x'_3 \in D$ for otherwise at least one vertex in $K_c$ would not be dominated. But then, $a_3$ is not dominated seeing that by Observation \ref{obs:ct3}, $a_3,b_3 \not\in D$ since $x_3 \in D$, and $t_3 \not\in D$ as shown previously. 

\begin{Claim}[\app]
\label{clm:ct3}
$\gamma (G_{\Phi}) = \vert X \vert + 4 \vert C \vert$ if and only if $ct_{\gamma} (G_{\Phi}) = 3$.
\end{Claim}

\begin{Claim}
\label{clm:phisat2}
$\gamma (G_{\Phi}) = \vert X \vert + 4 \vert C \vert$ if and only if $\Phi$ is satisfiable.
\end{Claim}

Assume first that $\gamma (G_{\Phi}) = \vert X \vert + 4 \vert C \vert$ and consider a minimum dominating set $D$ of $G_{\Phi}$. We construct a truth assignment from $D$ satisfying $\Phi$ as follows. For any $x \in X$, if $T_x \in D$, set $x$ to true; otherwise, set $x$ to false. We claim that each clause $c \in C$ has exactly one true variable. Indeed, it follows from Observation~\ref{obs:S} that $\vert D \cap V(G_c) \vert = 4$ for any $c \in C$, and from Claim \ref{clm:ct3} that $ct_{\gamma} (G_{\Phi}) = 3$. But then, by Observation \ref{obs:Snotdom}, for any $i \in \{1,2,3\}$, $a_i \not\in D$ if and only if $b_i \in D$ ($a_i$ would otherwise not be dominated). It then follows from Observations \ref{obs:ai} and \ref{obs:bi} that $\vert D\cap \{a_1,a_2,a_3\}\vert = 2$ and $\vert D\cap \{b_1,b_2,b_3\}\vert = 1$ for any $c \in C$; but by Observation \ref{obs:aifx} we conclude that  $b_i \in D$ if and only if $T_{x_i} \in D$, which proves our claim.

Conversely, assume that $\Phi$ is satisfiable and consider a truth assignment satisfying $\Phi$. We construct a dominating set $D$ of $G_{\Phi}$ as follows. If variable $x$ is set to true, we add $T_x$ to $D$; otherwise, we add $F_x$ to $D$. For any clause $c \in C$ and $i \in \{1,2,3\}$, if $T_{x_i} \in D$, then add $b_i$ to $D$; otherwise, add $a_i$ to $D$. Since every clause has exactly one true variable, it follows that $\vert D \cap \{b_1,b_2,b_3\} \vert = 1$ and $\vert D \cap \{a_1,a_2,a_3\} \vert = 2$; finally add $v_{\{x_i\}}$ to $D$ where $b_i \in D$. Now clearly $\vert D \cap V(G_c) \vert = 4$ and every vertex in $G_c$ is dominated. Thus, $\vert D \vert =  \vert X \vert + 4 \vert C \vert$ and so by Observation \ref{obs:S}, $ \gamma (G_{\Phi})= \vert X \vert + 4 \vert C \vert$, which concludes this proof. \\

Now combining Claims \ref{clm:ct3} and \ref{clm:phisat2}, we have that $\Phi$ is satisfiable if and only if $ct_{\gamma} (G_{\Phi}) = 3$ which completes the proof of Theorem \ref{thm:ct3}.
\qed
\end{proof}

By observing that for any graph $G$, $G$ is a \yes-instance for {\sc Contraction Number($\gamma$,$3$)} if and only if $G$ is a \no-instance for $2$-\textsc{Edge Contraction($\gamma$)}, we deduce the following corollary from Theorem \ref{thm:ct3}.

\begin{corollary}
$2$-\textsc{Edge Contraction($\gamma$)} is $\mathsf{coNP}$-hard.
\end{corollary}

It is thus $\mathsf{coNP}$-hard to decide whether $ct_{\gamma} (G) \leq 2$ for a graph $G$; and in fact, it is $\mathsf{NP}$-hard to decide whether equality holds, as stated in the following.

\begin{theorem}[\app]
\label{thm:gamma2}
{\sc Contraction Number($\gamma$,$2$)} is $\mathsf{NP}$-hard.
\end{theorem}

We finally consider the case $k=1$.

\begin{theorem}
\label{theorem:1econtracd}
$1$-\textsc{Edge Contraction($\gamma$)} is $\mathsf{NP}$-hard even when restricted to $P_t$-free graphs, with $t \geq 9$.
\end{theorem}

\begin{proof}
We reduce from \dom{}: given an instance $(G,\ell)$ of this problem, we construct an equivalent instance $G'$ of $1$-\textsc{Edge Contraction($\gamma$)} as follows. We denote by $\{v_1,\ldots,v_n\}$ the vertex set of $G$. The graph $G'$ consists of $\ell + 1$ copies of $G$, denoted by $G_0, \ldots, G_{\ell}$, connected in such a way that for any $1 \leq i \leq \ell$ and $1 \leq k \leq n$, the copies $v_k^i \in V(G_i)$ and $v_k^0 \in V(G_0)$ of a vertex $v_k$ of $G$ are true twins in the subgraph of $G'$ induced by $V(G_0) \cup V(G_i)$; and for any $1 \leq i,j \leq \ell$ and $1 \leq k \leq n$, the copies $v_k^i \in V(G_i)$ and $v_k^j \in V(G_j)$ of a vertex $v_k$ of $G$ are false twins in the subgraph of $G'$ induced by $\bigcup_{1 \leq p \leq \ell} V(G_p)$. Next, we add $\ell + 1$ pairwise nonadjacent vertices $x_1, \ldots, x_{\ell+1}$, which are made adjacent to every vertex in $G_0$; $x_i$ is further made adjacent to every vertex in $G_i$, for all $1 \leq i \leq \ell$. Finally, we add a vertex $y$ adjacent to only $x_{\ell +1}$ (see Fig. \ref{fig:dom1ec}). Note that the fact that for all $1\leq k \leq n$ and $1 \leq i,j \leq \ell$, $v_k^i$ and $v_k^j$ (resp. $v_k^i$ and $v_k^0$) are false (resp. true) twins within the graph induced by $\bigcup_{1 \leq p \leq \ell} V(G_p)$ (resp. $V(G_0) \cup V(G_i)$) is not made explicit on Fig. \ref{fig:dom1ec} for the sake of readability. In the following, we denote by $X= \{x_1,\ldots,x_{\ell + 1}\}$ and $V = \bigcup_{0 \leq p \leq \ell} V(G_p)$. We now claim the following.

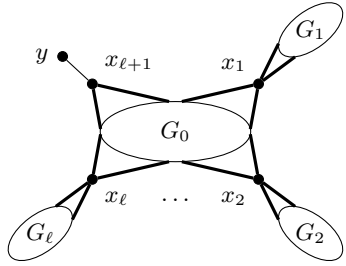
\begin{figure}[htb]
\centering
\begin{tikzpicture}[scale=.5]
\draw (0,0) ellipse (2cm and .8cm);
\node[draw=none] at (0,0) {$G_0$};

\node[circ,label=above left:{\small $x_1$}] (x1) at (2.21,1.27) {};
\draw[rotate=37] (4.5,0) ellipse (1cm and .5cm);
\node[draw=none] at (3.59,2.71) {$G_1$};
\draw[very thick] (x1) -- (1.99,.07)
(x1) -- (.17,.8)
(x1) -- (2.74,2.34)
(x1) -- (3.18,1.99);

\node[circ,label=below left:{\small $x_2$}] (x2) at (2.21,-1.27) {};
\draw[rotate=-37] (4.5,0) ellipse (1cm and .5cm);
\node[draw=none] at (3.59,-2.71) {$G_2$};
\draw[very thick] (x2) -- (1.99,-.07)
(x2) -- (.17,-.8)
(x2) -- (2.74,-2.34)
(x2) -- (3.18,-1.99);

\node[draw=none] at (0,-1.8) {$\ldots$};

\node[circ,label=below right:{\small $x_{\ell}$}] (xl) at (-2.21,-1.27) {};
\draw[rotate=37] (-4.5,0) ellipse (1cm and .5cm);
\node[draw=none] at (-3.59,-2.71) {$G_{\ell}$};
\draw[very thick] (xl) -- (-1.99,-.07)
(xl) -- (-.17,-.8)
(xl) -- (-2.74,-2.34)
(xl) -- (-3.18,-1.99);

\node[circ,label=above right:{\small $x_{\ell + 1}$}] (xl1) at (-2.21,1.27) {};
\node[circ,label=left:{\small $y$}] (y) at (-3,2) {};
\draw[very thick] (xl1) -- (-1.99,.07)
(xl1) -- (-.17,.8);
\draw[-] (xl1) -- (y);
\end{tikzpicture}
\caption{The graph $G'$ (thick lines indicate that the vertex $x_i$ is adjacent to every vertex in $G_0$ and $G_i$, for $i=1,\ldots,\ell + 1$).}
\label{fig:dom1ec}
\end{figure}

\begin{Claim}
\label{claim:gammaG'}
$\gamma (G') = \min \{\gamma (G) + 1, \ell + 1\}$. 
\end{Claim}

It is clear that $\{x_1, \ldots, x_{\ell + 1}\}$ is a dominating set of $G'$; thus, $\gamma (G') \leq \ell +1$. If $\gamma (G) \leq \ell$ and $\{v_{i_1},\ldots, v_{i_k}\}$ is a minimum dominating set of $G$, it is easily seen that $\{v^0_{i_1},\ldots, v^0_{i_k},x_{\ell +1}\}$ is a dominating set of $G'$. Thus, $\gamma (G') \leq \gamma (G) + 1$ and so, $\gamma (G') \leq \min \{\gamma (G) + 1, \ell + 1\}$. Now, suppose to the contrary that $\gamma (G') < \min \{\gamma (G) + 1, \ell + 1\}$ and consider a minimum dominating $D'$ set of $G'$. We first make the following simple observation.

\begin{observation}
\label{obs:yxl+1}
For any minimum dominating set $D$ of $G'$, $D \cap \{y,x_{\ell + 1 }\} \neq~\emptyset$.
\end{observation}

Now, since $\gamma (G') < \ell + 1$, there exists $1 \leq i \leq \ell$ such that $x_i \not\in D'$ (otherwise, $\{x_1,\ldots,x_{\ell}\} \subset D'$ and combined with Observation \ref{obs:yxl+1}, $D'$ would be of size at least $\ell + 1$). But then, $D'' = D' \cap V$ must dominate every vertex in $G_i$, and so $\vert D'' \vert \geq \gamma(G)$. Since $\vert D'' \vert \leq \vert D' \vert - 1$ (recall that $D' \cap \{y,x_{\ell +1}\} \neq \emptyset$), we then have $\gamma (G) \leq \vert D' \vert - 1$, a contradiction. Thus, $\gamma (G') = \min \{\gamma (G) + 1, \ell + 1\}$.

We now show that $(G,\ell)$ is a \yes-instance for \dom{} if and only if $G'$ is a \yes-instance for \contracd{}.

First assume that $\gamma (G) \leq \ell$. Then, $\gamma (G') = \gamma(G) + 1$ by the previous claim, and if $\{v_{i_1},\ldots,v_{i_k}\}$ is a minimum dominating set of $G$, then $\{v^0_{i_1},\ldots,v^0_{i_k},x_{\ell + 1}\}$ is a minimum dominating set of $G'$ which is not stable. Hence, by Theorem \ref{theorem:contracdom}(i), $G'$ is a \yes-instance for \contracd{}.

Conversely, assume that $G'$ is a \yes-instance for \contracd{} i.e., there exists a minimum dominating set $D'$ of $G'$ which is not stable (see Theorem \ref{theorem:contracdom}(i)). Then, Observation \ref{obs:yxl+1} implies that there exists $1 \leq i \leq \ell$ such that $x_i \not\in D'$; indeed, if it weren't the case, then by Claim \ref{claim:gammaG'} we would have $\gamma (G) = \ell + 1$ and thus, $D'$ would consist of $x_1, \ldots, x_{\ell}$ and either $y$ or $x_{\ell +1}$. In both cases, $D'$ would be stable, a contradiction. It follows that $D'' = D' \cap V$ must dominate every vertex in $G_i$ and thus, $\vert D'' \vert \geq \gamma (G)$. But $\vert D'' \vert \leq \vert D' \vert -1$ (recall that $D' \cap \{y,x_{\ell + 1}\} \neq \emptyset$) and so by Claim \ref{claim:gammaG'}, $\gamma (G) \leq \vert D' \vert - 1 \leq (\ell + 1) - 1$ that is, $(G,\ell)$ is a \yes-instance for \dom{}.

Finally, we can prove that if $G$ is $2K_2$-free then $G'$ is $P_9$-free. However, due to lack of space, this proof has been placed in Section \ref{section:P9} of the Appendix.
\qed
\end{proof}

Given the $\mathsf{NP}$-hardness of \contracd{} and its close relation to \dom{}, it is natural to consider the complexity of the problem when parameterized by the size of a minimum dominating set of the input graph. In the following, we denote by $mimw$ the mim-width parameter, and show that \contracd{} is $\mathsf{W}[1]$-hard when parameterized by $\gamma + mimw$. We first state two simple facts regarding the mim-width parameter.

\begin{observation}\label{obs:mimtwins}
Let $G$ be a graph and $u,v\in V(G)$ be two vertices that are true (resp. false) twins in $G$. Then $mimw(G-v)=mimw(G)$.
\end{observation}

\begin{observation}\label{obs:mimaddv}
Let $G$ be a graph and $v\in V(G)$. Then $mimw(G)\leq mimw(G-v)+1$.
\end{observation}

\begin{theorem}\label{thm:w1contracd}
\contracd{} is $\mathsf{W}[1]$-hard parameterized by $\gamma + mimw$.
\end{theorem}

\begin{proof}
We give a parameterized reduction from \dom{} parameterized by solution size plus mim-width, which is a problem that was recently shown to be $\mathsf{W}[1]$-hard by Fomin et al.~\cite{FOMIN18}. Given an instance $(G,\ell)$ of \dom{}, the construction of the equivalent instance $G'$ for \contracd{} is the same as the one introduced in the proof of Theorem~\ref{theorem:1econtracd}; and it is there shown that $G$ is a \yes-instance for \dom{} if and only if $G'$ is a \yes-instance for \contracd{}. Now, note that $G'$ can be obtained from $G$ by the addition of true twins (the set $V(G_1)$), the addition of false twins (the sets $V(G_2),\ldots,V(G_\ell)$), and the addition of $\ell+2$ vertices ($x_1,\ldots,x_{l+1},y$). By Observation~\ref{obs:mimtwins}, the addition of true (resp. false) twins does not increase the mim-width of a graph and, by Observation~\ref{obs:mimaddv}, the addition of a vertex can only increase the mim-width of $G$ by one; thus, $mimw(G')\leq mimw(G)+\ell+2$ and since $\gamma(G')\leq \ell+1$ by Claim~\ref{claim:gammaG'}, we conclude that $mimw(G')+\gamma(G')\leq mimw(G)+2\ell+3$.
\qed
\end{proof}

In order to obtain complexity results for further graph classes, let us now consider subdivisions of edges. 

\begin{lemma}[\app]
\label{lemma:3sub}
Let $G$ be a graph and let $G'$ be the graph obtained by 3-subdividing every edge of $G$. Then $G$ is a \yes-instance for \contracd{} if and only if $G'$ is a \yes-instance for \contracd{}.
\end{lemma}

By 3-subdividing every edge of a graph $G$ sufficiently many times, we deduce the following two corollaries from Lemma \ref{lemma:3sub}.

\begin{corollary}
\contracd{} is $\mathsf{NP}$-hard when restricted to bipartite graphs.
\end{corollary}

\begin{corollary} \label{theorem:largegirth}
For any $\ell\geq 3$, \contracd{} is $\mathsf{NP}$-hard when restricted to $\{C_3,\ldots,C_\ell\}$-free graphs.
\end{corollary}

We finally observe that, even if an edge is given, deciding whether contracting this particular edge decreases the domination number is unlikely to be solvable in polynomial time as shown in the following result.

\begin{theorem}[\app]
\label{thm:nop}
There exists no polynomial-time algorithm deciding whether contracting a given edge decreases the domination number, unless $\mathsf{P}=\mathsf{NP}$.
\end{theorem}


\section{Algorithms}
\label{s-poly}

We now deal with cases in which \kcontracd{} is tractable, for $k=1,2$. A first simple approach to the problem, from which Proposition \ref{prop:boundeddom} readily follows, is based on brute force.

\begin{proposition}[\app]
\label{prop:boundeddom}
For $k=1,2$, \kcontracd{} can be solved in polynomial time for a graph class $\mathcal{C}$, if either
\begin{itemize}
\item[(a)] $\mathcal{C}$ is closed under edge contractions and \dom{} can be solved in polynomial time for $\mathcal{C}$; or
\item[(b)] for every $G\in\mathcal{C}$, $\gamma(G)\leq q$, where $q$ is some fixed contant; or 
\item[(c)] $\mathcal{C}$ is the class of $(H+K_1)$-free graphs, where $|V_H|=q$ is a fixed constant and \kcontracd{} is polynomial-time solvable on $H$-free graphs.
\end{itemize}
\end{proposition}

Proposition \ref{prop:boundeddom}(b) provides an algorithm for \contracd{} parameterized by the size of a minimum dominating set of the input graph running in $\mathsf{XP}$-time. Note that this result is optimal as \contracd{} is $\mathsf{W}$[1]-hard with such parameterization from Theorem~\ref{thm:w1contracd}.

We further show that even though simple, this brute force method provides polynomial-time algorithms for a number of relevant classes of graphs, such as graphs of bounded tree-width and graphs of bounded mim-width. We first state the following result and observation.

\begin{theorem}{\normalfont \cite{VATSHELLE}}\label{thm:PCwidth}
Given a graph $G$ and a decomposition of width $t$, \dom{} can be solved in time $\Oh^*(3^t)$ when parameterized by tree-width, and in time $\Oh^*(n^{3t})$ when parameterized by mim-width.
\end{theorem}

\begin{observation}[\app]
\label{obs:mimwcontrac}
$mimw(G\setminus e)\leq mimw(G)+1$.
\end{observation}

\begin{proposition}\label{prop:xpmim}
Given a decomposition of width $t$, \kcontracd{} can be solved in time $\Oh^*(3^t)$ in graphs of tree-width at most $t$ and in time $\Oh^*(n^{3t})$ in graphs of mim-width at most $t$, for $k=1,2$.  
\end{proposition}

\begin{proof}
We use the above-mentioned brute force approach and Theorem~\ref{thm:PCwidth}. That is, for $k=1$, the algorithm first computes $\gamma(G)$ and then computes $\gamma(G\setminus e)$ for every $e \in E(G)$. For $k=2$, the algorithm proceeds similarly for every pair of edges. We next show that the width parameters increase by a constant when contracting at most two edges. It is a well-known fact that $tw(G\setminus e)\leq tw(G)$ and so, $tw(G\setminus \{e,f\})\leq tw(G)$. By Observation~\ref{obs:mimwcontrac}, $mimw(G\setminus e)\leq mimw(G)+1$ which implies that $mimw(G\setminus \{e,f\})\leq mimw(G)+2$. Also note that, given a tree (resp.\ mim) decomposition of width $t$ for $G$, we can construct in polynomial time decompositions of width $t$ (resp.\ at most $t+2$) for $G\setminus e$ and $G\setminus \{e,f\}$. This implies that $\gamma(G\setminus e)$ and $\gamma(G\setminus \{e,f\})$ can also be computed in time $\Oh^*(3^t)$ if $G$ is a graph of tree-width at most $t$, and in time $\Oh^*(n^{3t})$ if $G$ is a graph of mim-width at most $t$.
\qed
\end{proof}

Proposition~\ref{prop:xpmim} provides an algorithm for \contracd{} parameterized by mim-width running in $\mathsf{XP}$-time; this result is optimal as \contracd{} is $\mathsf{W}$[1]-hard parameterized by mim-width from Theorem~\ref{thm:w1contracd}. 

Since \dom{} is polynomial-time solvable in $P_4$-free graphs (see \cite{haynes}), it follows from Proposition \ref{prop:boundeddom}(a) that \kcontracd{} can also be solved efficiently in this graph class. However, \dom{} is $\mathsf{NP}$-complete for $P_5$-free graphs (see \cite{splitdom}) and thus, it is natural to examine the complexity of \kcontracd{} for this graph class. As we next show, \kcontracd{} is in fact polynomial-time solvable on $P_5$-free graphs, for $k~=~1,2$. 

\begin{lemma}
\label{lemma:p5free}
If $G$ is a $P_5$-free graph with $\gamma (G) \geq 3$, then $ct_{\gamma} (G) = 1$.
\end{lemma}

\begin{proof}
Let $G=(V,E)$ be a $P_5$-free graph and $D$ be a minimum dominating set of $G$. Suppose that $D$ is a stable set and consider $u,v \in D$ such that $d(u,v) = \max_{x,y \in D} d(x,y)$. Since $G$ is $P_5$-free, $d(u,v) \leq 3$ and, since $D$ is stable, $d(u,v) \geq 2$. We distinguish two cases depending on this distance.

\noindent{\bf Case 1.} $d(u,v) = 3$. Let $x$ (resp. $y$) be the neighbor of $u$ (resp. $v$) on a shortest path from $u$ to $v$. Then, $N(u) \cup N(v) \subseteq N(x) \cup N(y)$; indeed, if $a$ is a neighbor of $u$, then $a$ is nonadjacent to $v$ (recall that $d(u,v) = 3$) and thus, $a$ is adjacent to either $x$ or $y$ for otherwise $a,u,x,y$ and $v$ would induce a $P_5$ in $G$. The same holds for any neighbor of $v$. Consequently, $(D\backslash \{u,v\}) \cup \{x,y\}$ is a minimum dominating set of $G$ which is not stable; the result then follows from Theorem~\ref{theorem:contracdom}(i).

\noindent{\bf Case 2.} $d(u,v) = 2$. Since $D$ is stable and $d(u,v) = \max_{x,y \in D} d(x,y) = 2$, it follows that every $w \in D\backslash \{u,v\}$ is at distance two from both $u$ and $v$. Let $x$ (resp. $y$) be the vertex on a shortest path from $u$ (resp. $v$) to some vertex $w \in D\backslash \{u,v\}$. 

Suppose first that $x=y$. If every private neighbor of $w$ with respect to $D$ is adjacent to $x$ then $(D \backslash \{w\}) \cup \{x\}$ is a minimum dominating set of $G$ which is not stable; the result then follows from Theorem \ref{theorem:contracdom}(i). We conclude similarly if every private neighbor of $u$ or $v$ with respect to $D$ is adjacent to $x$. Thus, we may assume that $w$ (resp. $u$; $v$) has a private neighbor $t$ (resp. $r$; $s$) with respect to $D$ which is nonadjacent to $x$. Since $G$ is $P_5$-free, it then follows that $r$, $s$ and $t$ are pairwise adjacent. But then, $t,r,u,x$ and $v$ induce a $P_5$, a contradiction.

Finally, suppose that $x \neq y$ (we may also assume that $uy,vx \not\in E$ as we otherwise fall back in the previous case). Then, $xy \in E$ for $u,x,w,y$ and $v$ would otherwise induce a $P_5$. Now, if $a$ is a private neighbor of $u$ with respect to $D$ then $a$ is adjacent to either $x$ or $y$ ($a,u,x,y$ and $v$ otherwise induce a $P_5$); we conclude similarly that any private neihbor of $v$ with respect to $D$ is adjacent to either $x$ or $y$. If $b$ is adjacent to both $u$ and $v$ but not $w$, then it is adjacent to $x$ (and $y$) as $v,b,u,x$ and $w$ ($u,b,v,y$ and $w$) would otherwise induce a $P_5$. But then, $(D \backslash \{u,v\}) \cup \{x,y\}$ is a minimum dominating set of $G$ which is not stable; thus, by Theorem \ref{theorem:contracdom}(i), $ct_{\gamma} (G) = 1$ which concludes the proof.
\qed
\end{proof}

\begin{theorem}
\label{thm:p5free}
\kcontracd{} is polynomial-time solvable on $P_5$-free graphs, for $k=1,2$.
\end{theorem}

\begin{proof}
If $G$ has a dominating vertex, then $G$ is clearly a \no-instance for both $k=1,2$. Now, for every $uv\in E(G)$, we check whether $\{u,v\}$ is a dominating set. If it is the case, then by Theorem~\ref{theorem:contracdom}(i), $G$ is a \yes-instance for \kcontracd{} for $k=1,2$. If no edge of $G$ is dominating, we consider all the pairs of nonadjacent vertices of $G$. If there exists such a pair dominating $G$ and $k=1$ then by Theorem~\ref{theorem:contracdom}(i), we have a \no-instance for \contracd{} since this implies that every minimum dominating set of $G$ is stable. For the case $k=2$, if $G$ has two nonadjacent vertices dominating $G$, we then consider all triples of vertices of $G$ to check whether there exists one which is dominating and contains at least two edges (see Theorem \ref{theorem:contracdom}(ii)). Finally, both for $k=1$ and $k=2$, if $G$ has no dominating set of size at most two, then by Lemma~\ref{lemma:p5free}, $G$ is a \yes-instance for \kcontracd{}.
\qed
\end{proof}


\section{Conclusion}

In this paper, we studied the \kcontracd{} problem and provided the first complexity results. In particular, we showed that $1$-\textsc{Edge Contraction($\gamma$)} is $\mathsf{NP}$-hard for $P_t$-free graphs, $t\geq 9$, but polynomial-time solvable for $P_5$-free graphs; it would be interesting to determine the complexity status for $P_{\ell}$-free graphs, for $\ell\in\{6,7,8\}$. Similarly, the complexity of $2$-\textsc{Edge Contraction($\gamma$)} for $P_t$-free graphs, with $t\geq 6$, remains an interesting open problem.


\newpage
\appendix
\section*{Missing Proofs}

Section~\ref{section:clmct3} contains the proof of Claim \ref{clm:ct3} in Theorem \ref{thm:ct3}.

\noindent
Section~\ref{app1} contains the proof of Theorem~\ref{thm:gamma2}.

\noindent
Section~\ref{section:P9} contains the end of the proof ot Theorem \ref{theorem:1econtracd}.

\noindent
Section~\ref{section:3sub} contains the proof of Lemma \ref{lemma:3sub}.

\noindent 
Section~\ref{section:nop} contains the proof of Theorem \ref{thm:nop}.

\noindent
Section~\ref{section:easy} contains the proof of Proposition \ref{prop:boundeddom}.

\noindent
Section~\ref{section:mimw} contains the proof of Observation \ref{obs:mimwcontrac}.\\

\section{The proof of Claim \ref{clm:ct3} in Theorem \ref{thm:ct3}}
\label{section:clmct3}

Assume that $\gamma (G_{\Phi}) = \vert X \vert + 4 \vert C \vert$ and consider a minimum dominating set $D$ of $G_{\Phi}$. We first show that $D$ is a stable set which would imply that $ct_{\gamma} (G_{\Phi}) > 1$ (see Theorem \ref{theorem:contracdom}(i)). First note that Observation \ref{obs:S} implies that $\vert D \cap V(G_x) \vert = 1$ and $\vert D \cap V(G_c) \vert =4$, for any variable $x\in X$ and any clause $c\in C$. It then follows from Observation \ref{obs:Snotdom} that no truth vertex is dominated by some vertex $t_i$ or $x'_i$ in some clause gadget $G_c$ with $i\in \{1,2,3\}$; in particular, this implies that there can exist no edge in $D$ having one endvertex in some gadget $G_x$ ($x \in X$) and the other in some gadget $G_c$ ($c \in C$). Hence, it is enough to show that for any $c\in C$, $D \cap V(G_c)$ is a stable set.

Now consider a clause gadget $G_c$. It follows from Observation \ref{obs:Snotdom} that if there exists $i \in \{1,2,3\}$ such that $a_i \not \in D$ then $b_i \in D$ since $a_i$ must be dominated (also note that by Observation \ref{obs:Snotdom}, if $a_i \in D$ then $b_i \not\in D$). Hence, for any $i\in \{1,2,3\}$, exactly one of $a_i$ and $b_i$ belongs to $D$. But then, by Observation \ref{obs:Snotdom} and since $\vert D \cap V(G_c) \vert =4$ , we immediately conclude that $D \cap V(G_c)$ is a stable set and so, $D$ is a stable set.

Now, suppose to the contrary that $ct_{\gamma} (G_{\Phi}) = 2$ i.e., there exists a dominating set $D$ of $G_{\Phi}$ of size $\gamma(G_{\Phi}) + 1$ containing two edges $e$ and $e'$ (see Theorem \ref{theorem:contracdom}(ii)). First assume that there exists $x \in X$ such that $\vert D \cap V(G_x) \vert = 2$. Then, for any $x' \neq x$, $\vert D \cap V(G_{x'}) \vert =1$; and for any $c \in C$, $\vert D \cap V(G_c) \vert = 4$ which by Observation \ref{obs:Snotdom} implies that $\{t_i,x'_i\} \cap D = \emptyset$ for any $i \in \{1,2,3\}$. Since as shown previously, $D \cap V(G_c)$ is then a stable set, it follows that $D$ contains at most one egde, a contradiction.

Thus, there must exist some $c \in C$ such that $\vert D \cap V(G_c) \vert = 5$. We then claim that $\{a_1,a_2,a_3\} \not \subset D$. Indeed, since $x_1,x_2,x_3,v_{\{x_1\}}, v_{\{x_2\}}$ and $v_{\{x_3\}}$ must be dominated, $D \cap V(K_c) \neq \emptyset$ (otherwise, at least three additional vertices of $G_c$ would be required to dominate $x_1$, $x_2$ and $x_3$), say $v_{\{x_1\}} \in D$ without loss of generality. But then, $\vert N[x_1] \cap D \vert =1$ as $x_1$ must be dominated and $\vert D \cap V(G_c) \vert =5$ and so, $D$ contains at most one edge. Therefore, there must exist $i \in \{1,2,3\}$ such that $a_i \not\in D$, say $a_1 \not\in D$ without loss of generality. Then, since $a_1$ must be dominated, either $t_1 \in D$ or $b_1 \in D$. 

Assume first that $t_1$ belongs to $D$ (note that $\{b_1,x_1\} \cap D \neq \emptyset$ by Observation \ref{obs:Si}). Then, it follows from Observation \ref{obs:Si} that either $e$ or $e'$ has an endvertex in $\{a_j,b_j,x_j\}$ for some $j \neq 1$, say $j=2$ without loss of generality. Suppose that $x_2$ is an endvertex of $e$. Then the other endvertex of $e$ should be $b_2$ for otherwise it belongs to $K_c$ and thus, $a_2$ would not be dominated. But then, we conclude by Observation \ref{obs:Si} and the fact that $\vert D \cap V(G_c) \vert = 5$, that $D$ contains only one edge. Thus, $e=a_2b_2$ or $e=a_2t_2$ and since $v_{\{x_1\}}$ must be dominated, necessarily $x_3 \in D$; but then, $a_3$ is not dominated. Therefore, it must be that $b_1$ belongs to $D$; and we conclude similarly that if $a_2$ (resp. $a_3$) is not in $D$ then $b_2$ (resp. $b_3$) belongs to $D$.

Now, since $t_1,a_1 \not \in D$, it follows that $T_{x_1} \in D$ for otherwise $t_1$ would not be dominated. But $\vert D \cap V(G_x) \vert = 1$ and so, $F_{x_1} \not\in D$; thus, $D \cap \{x'_1,v_{\{x_1\}}\} \neq \emptyset$ as $x'_1$ must be dominated and we may assume, without loss of generality, that in fact, $v_{\{x_1\}} \in D$. Then, if $D \cap \{v_{\{x_2\}},v_{\{x_3\}}\} = \emptyset$, necessarily $F_{x_2},F_{x_3} \in D$; indeed, since $\vert D \cap V(G_c) \vert = 5$, at least one among $x'_2$ and $x'_3$ does not belong to $D$, say $x'_2$ without loss of generality. But if $x'_3 \in D$, then exactly one of $a_j$ and $b_j$, for $j \neq 1$ belongs to $D$ (recall that if $a_j \not\in D$ then $b_j \in D$) and therefore, $D$ contains at most one edge. Thus, $F_{x_2},F_{x_3} \in D$ which implies that $D \cap \{t_j,a_j\} \neq \emptyset$ for $j\neq 1$ as $t_j$ must be dominated. But by Observation \ref{obs:Si} and the fact that $\vert D \cap V(G_c) \vert = 5$, we have that $\vert D \cap \{t_2,t_3\} \vert \leq 1$ and so, $D$ contains at most one edge. Thus, $D \cap \{v_{\{x_2\}},v_{\{x_3\}}\} \neq \emptyset$ and since by Observation \ref{obs:Si} $\vert D \cap V(K_c) \vert \leq 2$, we conclude that in fact $\vert D \cap V(K_c) \vert = 2$. But then, exactly one among $a_j$ and $b_j$ belongs to $D$ for $j \neq 1$ and so, $D$ contains only one edge. Consequently, no such dominating set $D$ exists and thus, $ct_{\gamma} (G_{\Phi}) = 3$.\\

Conversely, assume that $ct_{\gamma}(G_{\Phi}) = 3$ and consider a minimum dominating set $D$ of $G_{\Phi}$. It readily follows from Observations \ref{obs:S} and \ref{obs:ct3} that for any variable $x \in X$, $\vert D \cap V(G_x) \vert = 1$. Now consider a clause gadget $G_c$. Then, by Observation \ref{obs:ct3}, we obtain that $t_i\not\in D$ (resp. $x'_i \not\in D$) for $i\in \{1,2,3\}$, as otherwise it would be within distance at most 2 from the vertex in $D$ belonging to the gadget $G_{x_i}$. 

Now since for any $i \in \{1,2,3\}$, $t_i \not \in D$, if $a_i \not\in D$ then $b_i \in D$ as $a_i$ must be dominated (also note that by Observation \ref{obs:ct3}, if $a_i \in D$ then $b_i \not\in D$. Thus, by Observations \ref{obs:ai} and \ref{obs:bi}, we conclude that for any clause gadget $G_c$, $\vert D \cap \{a_1,a_2,a_3\} \vert = 2$ and $\vert D \cap \{b_1,b_2,b_3\} \vert = 1$, say $a_1,a_2,b_3 \in D$ without loss of generality. But then, $v_{\{x_3\}}$ must belong to $D$; indeed, since $b_3 \in D$, it follows that $T_{x_3} \in D$ for otherwise $t_3$ is not dominated. Observation \ref{obs:ct3} then implies that $x'_3 \not\in D$ and thus, it can only be dominated by $v_{\{x_3\}}$. But then, it follows from Observation \ref{obs:aifx} that every vertex in $G_c$ is dominated and we conclude that $\vert D \cap V(G_c) \vert =4$ by minimality of $D$. Consequently, $\vert D \vert = \vert X \vert + 4 \vert C \vert$ which concludes the proof of Claim \ref{clm:ct3}.

\section{The Proof of Theorem \ref{thm:gamma2}}
\label{app1}

The reduction is based on the following problem, which was shown to be $\mathsf{NP}$-complete by Dahlhaus et al. \cite{dahl}. 

\begin{center}
\begin{boxedminipage}{.99\textwidth}
\textsc{Exactly 3-Bounded 3-Sat}\\[2pt]
\begin{tabular}{ r p{0.8\textwidth}}
\textit{~~~~Instance:} &A formula $\Phi$ with variable set $X$ and clause set $C$ such that each variable has exactly three literals, with one of them occuring in two clauses and the other one in one, and each clause is the disjunction of two or three literals.\\
\textit{Question:} &Is $\Phi$ satisfiable?
\end{tabular}
\end{boxedminipage}
\end{center}

We reduce from \textsc{Exactly 3-Bounded 3-Sat}: given an instance $\Phi$ of this problem, with variable set $X$ and clause set $C$, we construct an equivalent instance $G_{\Phi}$ of {\sc Contraction Number($\gamma$,$2$)} as follows. First note that we may assume that $\vert X \vert \geq 4$ as \textsc{Exactly 3-Bounded 3-Sat} is otherwise polynomial-time solvable. The graph $G_{\Phi}$ then contains a copie of the graph $H$ depicted in Fig. \ref{fig:basegad}. For any variable $x \in X$, we introduce the gadget $G_x$ which has two distinguished \textit{literal vertices} $x$ and $\overline{x}$, as depicted in Fig. \ref{fig:vargad}. For any clause $c \in C$, we introduce a copie of $K_2$ with a distinguished \textit{clause vertex} $c$ and a distinguished \textit{transmitter vertex} $t_c$. Finally, for each clause $c \in C$, we add an edge between the clause vertex $c$ and the literal vertices whose corresponding literals belong to $c$; furthermore, we add an edge between the transmitter vertex $t_c$ and vertices $1$ and $3$ of the graph $H$. We first prove the following.

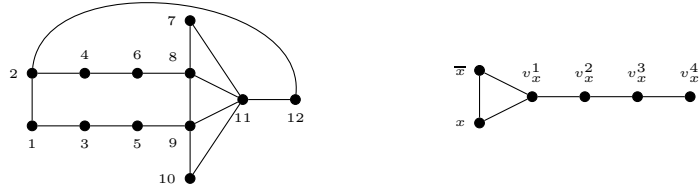
\begin{figure}[htb]
\centering
\begin{minipage}{0.45\linewidth}
\begin{subfigure}{1\linewidth}
\centering
\begin{tikzpicture}[node distance=.7cm]
\node[circ,label=below:{\tiny $1$}] (1) at (0,0) {};
\node[circ,label=left:{\tiny $2$},above of=1] (2) {};
\node[circ,label=below:{\tiny $3$},right of=1] (3) {};
\node[circ,label=below:{\tiny $5$},right of=3] (5) {};
\node[circ,label=below left:{\tiny $9$},right of=5] (9) {};
\node[circ,label=above:{\tiny $4$},right of=2] (4) {};
\node[circ,label=above:{\tiny $6$},right of=4] (6) {};
\node[circ,label=above left:{\tiny $8$},right of=6] (8) {};
\node[circ,label=left:{\tiny $7$},above of=8] (7) {};
\node[circ,label=left:{\tiny $10$},below of=9] (10) {};
\node[circ,label=below:{\tiny $11$}] (11) at (2.8,.35) {};
\node[circ,label=below:{\tiny $12$},right of=11] (12) {};

\draw[-] (1) -- (2)
(1) -- (3)
(2) -- (4)
(3) -- (5)
(4) -- (6)
(5) -- (9)
(6) -- (8)
(7) -- (8) 
(7) -- (11)
(8) -- (9)
(8) -- (11)
(9) -- (10)
(9) -- (11)
(10) -- (11)
(11) -- (12);
\draw (2) edge[bend left=90] (12);
\end{tikzpicture}
\caption{The graph $H$.}
\label{fig:basegad}
\end{subfigure}
\end{minipage}
\begin{minipage}{.45\linewidth}
\begin{subfigure}{1\linewidth}
\centering
\begin{tikzpicture}[node distance=.7cm]
\node[circ,label=left:{\tiny $x$}] (x) at (0,0) {};
\node[circ,label=left:{\tiny $\overline{x}$},above of=x] (xbar) {};
\node[circ,label=above:{\tiny $v_x^1$}] (v1) at (.7,.35) {};
\node[circ,label=above:{\tiny $v_x^2$},right of=v1] (v2) {};
\node[circ,label=above:{\tiny $v_x^3$},right of=v2] (v3) {};
\node[circ,label=above:{\tiny $v_x^4$},right of=v3] (v4) {};

\draw (x) -- (xbar)
(x) -- (v1)
(xbar) -- (v1)
(v1) -- (v2) 
(v2) -- (v3)
(v3) -- (v4);

\node[draw=none] at (0,2) {};
\node[draw=none] at (0,-.8) {};
\end{tikzpicture}
\caption{The gadget $G_x$ with $x \in X$.}
\label{fig:vargad}
\end{subfigure}
\end{minipage}
\caption{Construction of the graph $G_{\Phi}$.}
\end{figure}

\begin{Claim}
\label{clm:H}
$\gamma (H) = \gamma (H - \{1,3\}) =3$ and $ct_{\gamma} (H) = 2$.
\end{Claim}

Since $\{3,4,11\}$ (resp. $\{4,5,11\}$) is a dominating set of $H$ (resp. $H- \{1,3\}$), it follows that $\gamma (H) \leq 3$ and $\gamma (H - \{1,3\}) \leq 3$. On the other hand, any dominating set of $H$ must contain at least three vertices as $\{3,4,11\}$ is a stable set with $N(3) \cap N(4) = N(3) \cap N(11) = N(4) \cap N(11) = \emptyset$. Similarly, any dominating set of $H - \{1,3\}$ must contain at least three vertices as $\{4,5,7\}$ is a stable set with $N(4) \cap N(5) = N(4) \cap N(11) = N(5) \cap N(11) = \emptyset$. Thus, $\gamma (H) = \gamma (H - \{1,3\}) = 3$.

We now claim that $H$ has a unique minimum dominating set, namely $\{3,4,11\}$. First observe that any minimum dominating $D$ set of $H$ contains vertex 11 as otherwise $D$ would have to contain at least two vertices from $\{7,8,9,10\}$ in order to dominate vertices $7$ and $10$, and at least two other vertices to dominate vertices $3$ and $4$; but then, $\vert D \vert \geq 4 > \gamma (H)$. Now if there exists a minimum dominating set $D$ not containing vertex $4$, then $\{2,6\} \cap D \neq \emptyset$ as vertex $4$ is dominated. But if $2 \in D$ then $\{6,8\} \cap D \neq \emptyset$ as $6$ must be dominated; and so, $\vert D \vert \geq 4$ as $11 \in D$ and $\{1,3,5\} \cap D \neq \emptyset$ ($3$ must be dominated). Otherwise, $6 \in S$ and similarly $\{2,12\} \cap D \neq \emptyset$ as $2$ must be dominated; and we conclude similarly that $\vert D \vert \geq 4$. Thus, every minimum dominating set contains vertex $4$; we conclude similarly that every minimum dominating set contains vertex $3$. It follows that $\{3,4,11\}$ is the only minimum dominating set of $H$ and since it is stable, we obtain that $ct_{\gamma} (H) > 1$. Now, $\{1,2,8,9\}$ is clearly dominating and since it contains two edges, it follows that $ct_{\gamma} (H) = 2$ (see Theorem \ref{theorem:contracdom}(ii)). This completes the proof of the claim.

We next prove two claims which together show that $\Phi$ is satisfiable if and only if $ct_{\gamma} (G_{\Phi}) =~2$.

\begin{Claim}
\label{clm:ct2}
$\gamma (G_{\Phi}) = 2 \vert X \vert + 3$ if and only if $ct_{\gamma} (G_{\Phi}) = 2$.
\end{Claim}

Suppose that $\gamma (G_{\Phi}) = 2 \vert X \vert + 3$ and let $D$ be a minimum dominating set of $G_{\Phi}$. Since for any $x \in X$, vertices $v_x^1$ and $v_x^4$ can only be dominated by (distinct) vertices in $V(G_x)$, it follows that $\vert D \cap V(G_x) \vert \geq 2$. Furthermore, $\vert D \cap V(H) \vert \geq 3$ as $\gamma (H) = 3$ by Claim~\ref{clm:H} and even if vertices $1$ and $3$ are dominated by some transmitter vertex, we still have $\gamma (H - \{1,3\}) = 3$ by Claim \ref{clm:H}. Now, since $\vert D \vert = 2 \vert X \vert + 3$ we have that:
\begin{itemize}
\item[$\cdot$] $\forall x \in X$, $\vert D \cap V(G_x) \vert =2$;
\item[$\cdot$] $\vert D \cap V(H) \vert =3$;
\item[$\cdot$] $\forall c \in C$, $D \cap V(G_c) = \emptyset$.
\end{itemize} 

\noindent
But then, for any $x \in X$, the set $D \cap V(G_x)$ is a minimum dominating set of $G_x$ and therefore stable as we trivially have $ct_{\gamma} (G_x) = 2$. Similarly, $D \cap V(H)$ is a dominating set of $H$ (recall that for any $c \in C$, $D \cap V(G_c) = \emptyset$) and therefore stable as $ct_{\gamma} (H) = 2$ by Claim \ref{clm:H}. Thus, $D$ is stable and since $(D \cap \bigcup_{x \in x} V(G_x)) \cup \{1,2,8,9\}$ is a dominating set of $G_{\Phi}$ of size $\gamma(G_{\Phi}) + 1$ containing two edges, it follows that $ct_{\gamma} (G_{\Phi}) = 2$.\\

Conversely, assume that $ct_{\gamma} (G_{\Phi}) = 2$ and let $D$ be a minimum dominating set of $G_{\Phi}$ (note that $D$ is stable). Suppose that there exists $c \in C$ such that $D \cap V(G_c) \neq \emptyset$. Then, we may assume that $t_c \in D$; indeed, if $c \in D$ then no literal vertex adjacent to $c$ is in the dominating set as $D$ is stable. We then claim that any literal vertex adjacent to $c$ must dominated by one of its neighbor in the gadget; if $x$ (or $\overline{x}$) is adjacent to $c$ and neither $v_x^1$ nor $\overline{x}$ belongs to $D$, then we necessarily have $\vert D \cap \{v_x^2,v_x^3,v_x^4\} \vert = 2$ and so, $G_{\Phi}$ would have a minimum dominating set which is not stable, namely $(D \backslash \{v_x^1,v_x^2,v_x^3,v_x^4\}) \cup \{v_x^2,v_x^3\}$, a contradiction. But then, $(D \backslash \{c\}) \cup \{t_c\}$ is a minimum dominating set of $G_{\Phi}$. Now since $t_c \in D$, it follows that $D \cap \{1,3\} = \emptyset$ as $D$ is stable, which implies that $\{c',t_{c'}\} \cap D \neq \emptyset$ for any $c' \in C$. In particular, the set $D' = (D \backslash \{t_{c'}, c' \neq c\}) \cup \{c', c' \neq c\}$ is a minimum dominating set of $G_{\Phi}$ and thus, stable. But $\vert X \vert \geq 4$ so there must exist $x \in X$ such that both $x$ and $\overline{x}$ are dominated in $D'$ by some clause vertices (take any variable $x$ not occuring in $c$). In particular, $\{x,\overline{x}\} \cap D' = \emptyset$ which implies that $\vert D' \cap \{v_x^1,v_x^2,v_x^3,v_x^4\} \vert = 2$; but then $(D' \backslash \{v_x^1,v_x^2,v_x^3,v_x^4\}) \cup \{v_x^2,v_x^3\}$ is a minimum dominating set of $G_{\Phi}$ which is not stable, a contradiction. It follows that for any $c \in C$, $D \cap V(G_c) = \emptyset$.

On the other hand, if there exists $x \in X$ such that $\vert D \cap V(G_x) \vert > 2$, it is not difficult to see that $D$ could then be transformed into a minimum dominating set which is not stable. But since for any $x \in X$, at least two vertices are required to dominate $\{v_x^1,v_x^2,v_x^3,v_x^4\}$, we have then that $\vert D \cap V(G_x) \vert = 2$. Finally, as $D \cap V(H)$ is a minimum dominating set of $H$ (recall that $D \cap V(G_c) = \emptyset$ and so, no vertex in $(V(G_{\Phi}) \backslash V(H)) \cap D$ dominates a vertex in $H$), $\vert D \cap V(H) \vert  = \gamma (H) = 3$. Thus, $\gamma (G_{\Phi}) = 2 \vert X \vert + 3$, which concludes the proof of the claim.

\begin{Claim}
\label{clm:phisat}
$\gamma (G_{\Phi}) = 2 \vert X \vert + 3$ if and only if $\Phi$ is satisfiable.
\end{Claim}

Assume first that $\gamma (G_{\Phi}) = 2 \vert X \vert + 3$ and consider a minimum dominating set $D$ of $G$. As shown in the proof of Claim \ref{clm:ct2}, $D$ is then stable and contains no vertex from $\bigcup_{c \in C} V(G_c)$. Therefore, any clause vertex is dominated by a literal vertex and for any $x \in X$, $ \vert D \cap \{x,\overline{x}\} \vert \leq 1$. We may thus construct a truth assignment which satisfies $\Phi$ as follows.
\begin{itemize}
\item[$\cdot$] If $x \in D$, set variable $x$ to true;
\item[$\cdot$] if $\overline{x} \in D$, set variable $x$ to false;
\item[$\cdot$] otherwise, we may set variable $x$ to any truth value.
\end{itemize}

Conversely, assume that $\Phi$ is satisfiable and consider a truth assignment which satisfies $\Phi$. We construct a dominating set $D$ of $G_{\Phi}$ as follows. For any $x \in X$, if $x$ is set to true, we add $x$ and $v_x^3$ to $D$, otherwise we add $\overline{x}$ and $v_x^3$ to $D$. We further add vertices $3$, $4$ and $11$ of $H$. Then, it is not difficult to see that $D$ is dominating (every transmitter vertex is dominated by vertex $3$ and every clause vertex has an adjacent literal vertex belonging to $D$) and so, $\gamma (G_{\Phi}) \leq 2 \vert X \vert + 3$. But since for any $x \in X$, $\vert D \cap V(G_x) \vert \geq 2$ and for any $c \in C$, $\vert D \cap V(G_C) \vert \geq 4$, it follows that $\gamma (G_{\Phi}) = 2 \vert X \vert + 3$. This completes the proof of the claim.

Now combining Claims \ref{clm:ct2} and \ref{clm:phisat}, we have that $\Phi$ is satisfiable if and only if $ct_{\gamma} (G_{\Phi}) = 2$ which concludes the proof of Theorem \ref{thm:gamma2}.

\section{The end of the proof of Theorem \ref{theorem:1econtracd}}
\label{section:P9}

We here prove that if $G$ is $2K_2$-free then $G'$ is $P_9$-free; since \dom{} is $\mathsf{NP}$-hard in $2K_2$-free graphs \cite{splitdom}, this would complete the proof of the theorem. Suppose to the contrary that $G'$ contains an induced path of length $9$ and consider such a path $P=z_1 - \ldots - z_9$. Observe first that there exist no $p,q \in \{1,\ldots, 8\}$ with $p + 1 < q-1$ and $i,j \in \{1,\ldots ,\ell\}$ such that $z_p,z_{q+1} \in V(G_i)$ and $z_{p+1},z_q \in V(G_j)$ as $G$ would otherwise contain a $2K_2$, namely $x,z,u$ and $v$ where $z_p,z_{p+1}, z_q$ and $z_{q+1}$ is the copy of $x,z,u$ and $v$, respectively. We now claim the following.

\begin{observation}
\label{clm:2k2}
If there exist $1 \leq i,j,k,l \leq \ell,$ with possibly $j=k$, such that $z_p \in V(G_i)$, $z_{p+1}\in V(G_j)$, $z_q \in V(G_k)$ and $z_{q+1} \in V(G_l)$ for some $p,q \in \{1,\ldots ,8\}$ with $p+1 < q-1$, then $P$ is not induced (see Fig. \ref{fig:notind}).
\end{observation}

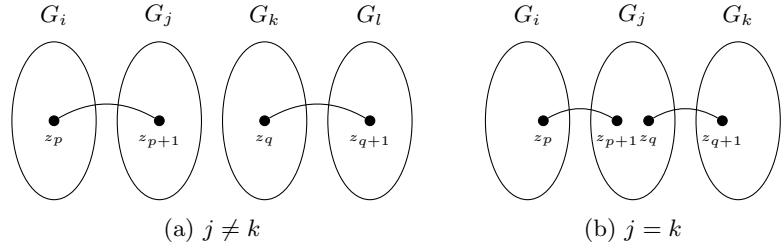
\begin{figure}[htb]
\centering
\begin{minipage}{.45\linewidth}
\begin{subfigure}{1\linewidth}
\centering
\begin{tikzpicture}[scale=.7]
\draw (0,0) ellipse (.8cm and 1.5cm);
\node[draw=none] at (0,2) {\small $G_i$};
\node[circ,label=below:{\tiny $z_p$}] (p) at (0,0) {};

\draw (2,0) ellipse (.8cm and 1.5cm);
\node[draw=none] at (2,2) {\small $G_j$};
\node[circ,label=below:{\tiny $z_{p+1}$}] (p1) at (2,0) {};

\draw (4,0) ellipse (.8cm and 1.5cm);
\node[draw=none] at (4,2) {\small $G_k$};
\node[circ,label=below:{\tiny $z_q$}] (q) at (4,0) {};

\draw (6,0) ellipse (.8cm and 1.5cm);
\node[draw=none] at (6,2) {\small $G_l$};
\node[circ,label=below:{\tiny $z_{q+1}$}] (q1) at (6,0) {};

\draw (p) edge[bend left] (p1)
(q) edge[bend left] (q1);
\end{tikzpicture}
\caption{$j \neq k$}
\end{subfigure}
\end{minipage}
\begin{minipage}{.45\linewidth}
\begin{subfigure}{1\linewidth}
\centering
\begin{tikzpicture}[scale=.7]
\draw (0,0) ellipse (.8cm and 1.5cm);
\node[draw=none] at (0,2) {\small $G_i$};
\node[circ,label=below:{\tiny $z_p$}] (p) at (.3,0) {};

\draw (2,0) ellipse (.8cm and 1.5cm);
\node[draw=none] at (2,2) {\small $G_j$};
\node[circ,label=below:{\tiny $z_{p+1}$}] (p1) at (1.7,0) {};
\node[circ,label=below:{\tiny $z_q$}] (q) at (2.3,0) {};

\draw (4,0) ellipse (.8cm and 1.5cm);
\node[draw=none] at (4,2) {\small $G_k$};
\node[circ,label=below:{\tiny $z_{q+1}$}] (q1) at (3.7,0) {};

\draw (p) edge[bend left] (p1)
(q) edge[bend left] (q1);
\end{tikzpicture}
\caption{$j = k$}
\end{subfigure}
\end{minipage}
\caption{Forbidden configurations for $P$.}
\label{fig:notind}
\end{figure}

Indeed, suppose that such indices exist. Then, since $z_pz_{p+1} \in E(G')$, $z_p$ and $z_{p+1}$ cannot be copies of the same vertex in $G$ (recall that copies in $G_i$ and $G_j$ of a given vertex are false twins within the graph induced by $V(G_i) \cup V(G_j)$); and the same holds for $z_q$ and $z_{q+1}$. Now, since $z_qz_{q+1} \in E(G')$ and $z_p$ is nonadjacent to $z_{q+1}$ (resp. $z_q$), it follows that $z_p$ and $z_q$ (resp. $z_{q+1}$) cannot be copies of the same vertex in $G$; we conclude similarly that $z_{p+1}$ and $z_q$ (resp. $z_{q+1}$) are not copies of the same vertex in $G$. Thus, $z_p$, $z_{p+1}$, $z_q$ and $z_{q+1}$ are copies of four distinct vertices in $G$, say $x$, $z$, $u$ and $v$ respectively. But then, $xz,uv \in E(G)$ with $x,z \not \in N(u) \cup N(v)$ as $P$ is induced and so, $x,z,u$ and $v$ induce a $2K_2$ in $G$, a contradiction.\\

We next prove the following observation. Note that it is symmetric in the sense that if there exists $5 < p \leq 9$ such that for any $q \leq p$, $z_q \not \in \{x_1,\ldots,x_{\ell + 1}\} \cup \{y\}$, then $G$ contains a $2K_2$. 

\begin{observation}
\label{clm:p5}
If there exists $1 \leq p \leq 5$ such that for any $q \geq p$, $z_q \not \in \{x_1,\ldots,x_{\ell + 1}\} \cup \{y\}$, then $G$ contains a $2K_2$.
\end{observation}

We show that if such an index exists, then $z_p \ldots z_9$ corresponds in $G$ to an induced path of length at least $9 - p+ 1 \geq 5$ and hence, $G$ contains a $2K_2$. To this end, note that it is sufficient to show that no two $z_i$ and $z_j$, with $i,j \in \{p,\ldots,9\}$, are copies of the same vertex in $G$; indeed, if $z_i$ (resp. $z_j$) is the copie of $v_i \in V(G)$ (resp. $v_j \in V(G)$) with $v_i \neq v_j$ for any $i,j \in \{p,\ldots,9\}$, then $v_p\ldots v_9$ induces a path of length at least 5 in $G$. So suppose to the contrary that there exists $i,j \in \{p,\ldots,9\}$ with $i < j$, such that $z_i$ and $z_j$ are copies of the same vertex in $G$. Then, either both have only one neighbor in $\{z_p,\ldots,z_9\}$ in which case $i=p$, $j=9$ and $z_i z_8 \in E(G')$ since $z_i$ and $z_j$ are false twins, and so, $P$ is not induced. Or, at least one among $z_i$ and $z_j$ has two neighbors in $\{z_p,\ldots,z_9\}$: if it is $z_i$ then $i > p$ and $z_jz_{i-1} \in E(G')$ since $z_i$ and $z_j$ are false twins and so, $P$ is not induced; otherwise $j < 9$ and $z_iz_{j+1} \in E(G')$ since $z_i$ and $z_j$ are false and again, $P$ is not induced. We therefore conclude that for any $i,j \in \{p,\ldots, 9\}$, $z_i$ and $z_j$ are not copies of the same vertex in $G$.\\

We now claim that $P$ contains exactly two vertices from $\{x_1,\ldots,x_{\ell + 1}\}$. Indeed, if $P$ contained $x_i$, $x_j$ and $x_k$ in this order, with $1 \leq i,j,k \leq \ell + 1$, then $P$ would contain no vertex from $G_0$ for otherwise $P$ would not be induced (recall that $x_i$, $x_j$ and $x_k$ are adjacent to every vertex in $G_0$). But then, there exists $p,q \in \{1,\ldots , 8\}$ such that $z_p \in V(G_i)$, $z_{p+1},z_q \in V(G_j)$ and $z_{q+1} \in V(G_k)$ and we conclude to a contradiction by Observation \ref{clm:2k2}. On the other hand, if $P$ contains no vertex from  $\{x_1,\ldots,x_{\ell + 1}\}$, then it does not contain $y$ either (otherwise $P$ would contain $x_{\ell + 1}$ as it is the only neighbor of $y$) and we conclude by Observation \ref{clm:p5} that $G$ contains a $2K_2$. Finally, if $P$ contains exactly one vertex $z_p$ ($1 \leq p \leq 9$) from $\{x_1,\ldots,x_{\ell + 1}\}$ then, either $p \leq 5$ in which case it is clear that $z_q \neq y$ for any $q \geq p$ and we conclude by Observation \ref{clm:p5}; or $p > 5$ and since $z_q \neq y$ for any $q \leq p$, we also conclude by Observation \ref{clm:p5}. It follows that $P$ must contain exactly two vertices from $\{x_1,\ldots, x_{\ell + 1}\}$, say $x_i$ and $x_j$ in this order.

Now, if the subpath $P'$ of $P$ from $x_i$ to $x_j$ intersects $G_0$ then $P'$ has length 3 (recall that $x_i$ and $x_j$ are adjacent to every vertex in $G_0$) and if $i \leq \ell$ (resp. $j \leq \ell$), $P$ contains at most one vertex from $G_i$ (resp. $G_j$) as $x_i$ and $x_j$ have a neighbor in $P$ belonging to $G_0$. But since $P$ has length 9, it follows $P$ contains exactly one vertex from at least one of $G_i$ an $G_j$, say $z_p \in G_j$ without loss of generality, and there must exist $1 \leq k \leq \ell$ with $k \neq i,j$ such that $ z_{p+1} \in G_k$ (see Fig. \ref{fig:2x}). Now if $p > 5$ then $x_i = z_q$ for some $q > 2$ and since $P$ contains at most one vertex from $G_i$, it follows that there exists $l \in \{1, \ldots, \ell \}$ with $l \neq i,j$, such that $z_{q-2} \in V(G_l)$. Using $z_{q-2},z_{q-1},z_p$ and $z_{p+1}$, we then conclude by Observation \ref{clm:2k2} that $P$ is not induced. Thus, $p \leq 5$ but then, it follows from Observation \ref{clm:p5} that $G$ contains a $2K_2$. Hence, $P'$ contains no vertex from $G_0$. 

\begin{figure}[htb]
\centering
\begin{tikzpicture}[scale=.7]
\node[circ,label=left:{\tiny $x_i$}] (xi) at (0,-2.5) {};

\draw (2,0) ellipse (.8cm and 1.5cm);
\node[draw=none] at (2,2) {\small $G_j$};
\node[circ,label=right:{\tiny $x_j$}] (xj) at (2,-2.5) {};
\node[circ,label=above:{\tiny $z_p$}] (zp) at (2,0) {};

\draw (4,0) ellipse (.8cm and 1.5cm);
\node[draw=none] at (4,2) {\small $G_k$};
\node[circ,label=above:{\tiny $z_{p+1}$}] (zp1) at (4,0) {};

\node[circ] (g0) at (1,-3.5) {};

\draw[dashed] (-1,-3) -- (5,-3) node[label=below:{\small $G_0$}] {};

\draw (zp) edge[bend left] (zp1);
\draw (xj) -- (zp)
(xi) -- (g0)
(g0) -- (xj);

\draw[thick] (xj) -- (1.43,-1.06)
(xj) -- (2.57,-1.06);
\end{tikzpicture}
\caption{$P$ contains two vertices from $\{x_1, \ldots, x_{\ell + 1}\}$.}
\label{fig:2x}
\end{figure}
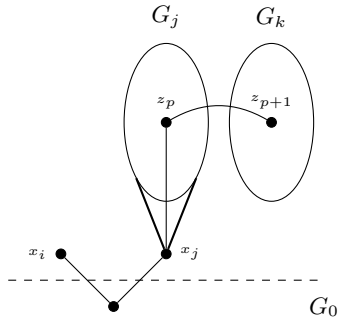 

Then, $x_i = z_q$ for some $q \in \{1,\ldots, 9\}$, and either $q > 2$ in which there exist $k,l \in \{1,\ldots, \ell\}$ with $k,l \neq i$, such that $z_{q-2} \in V(G_k)$ and $z_{q+2} \in V(G_l)$ as $P$ then contains exactly two vertices from $G_i$ (note that it might be that $l =j$, but necessarily, $l \neq k$); we then conclude by Observation \ref{clm:2k2} that $G$ contains a $2K_2$. Or $q \leq 2$ which implies that $z_{q+2} \in V(G_k)$ for some $k \neq i$. Then, $x_j = z_p$ and either $p < 8$ in which case there exists $1 \leq l \leq \ell$ with $l \neq j$, such that $z_{p+2} \in V(G_l)$; we then conclude by Observation \ref{clm:2k2} by considering $z_{q+1} \in V(G_i)$, $z_{q+2} \in V(G_k)$, $z_{p+1} \in V(G_j)$ and $z_{p+2} \in V(G_l)$ (note that it might be that $k = l$ or $k=j$). Or $p \geq 8$ and there exists $1 \leq l \leq \ell$ with $l \neq j$, such that $z_{p-2} \in V(G_l)$; but again, we conclude by Observation \ref{clm:2k2} by considering $z_{q+1} \in V(G_i)$, $z_{q+2} \in V(G_k)$, $z_{p-1} \in V(G_j)$ and $z_{p-2} \in V(G_l)$ (note that it might be that $k=l$ or $k=j$). Thus, it follows that no such path exists i.e., $G'$ is $P_9$-free which concludes the proof.

\section{The proof of Lemma \ref{lemma:3sub}}
\label{section:3sub}

Let $G=(V,E)$ be a graph. In the following, given an edge $e \in E$, we denote by $e_1$, $e_2$ and $e_3$ the three new vertices resulting from the 3-subdivision of the edge $e$. We first prove the following.

\begin{Claim}
\label{claim:3sub}
If $H$ is the graph obtained from $G$ by 3-subdividing one edge, then $\gamma (H) = \gamma (G) + 1$.
\end{Claim}

Assume that $H$ is obtained by 3-subdividing the edge $e=uv$ (we assume in the following that $e_1$ is adjacent to $u$ and $e_3$ is adjacent to $v$ in $H$), and consider a minimum dominating set $D$ of $G$. We construct a dominating set of $H$ as follows. If $D \cap \{u,v\} = \emptyset$, then $D \cup \{e_2\}$ is a dominating set of $H$. If $|D \cap \{u,v\}| = 1$, then we may assume without loss of generality that $u \in D$; but then, $D \cup \{e_3\}$ is a dominating set of $H$. Finally, if $\{u,v\} \subset D$, then $D \cup \{e_1\}$ is a dominating set of $H$. Thus, $\gamma (H) \leq \gamma (G) + 1$.

\begin{figure}[htb]
\centering
\begin{tikzpicture}[node distance=.5cm,scale=.7]
\node[circ,label=above:{\tiny $u$}] (u1) at (0,0) {};
\node[circ,label=above:{\tiny $v$}] (v1) at (1.5,0) {};

\draw[-Implies,line width=.6pt,double distance=2pt] (2.5,0) -- (3.5,0);

\node[circ,label=above:{\tiny $u$}] (u'1) at (4.5,0) {};
\node[circ,right of=u'1,label=above:{\tiny $e_1$}] (e1) {}; 
\node[circr,right of=e1,label=above:{\tiny $e_2$}] (e2) {};
\node[circ,right of=e2,label=above:{\tiny $e_3$}] (e3) {};
\node[circ,right of=e3,label=above:{\tiny $v$}] (v'1) {};

\draw[-] (u1) -- (v1) node[above,midway] {\tiny $e$}
(u'1) -- (e1)
(e1) -- (e2)
(e2) -- (e3)
(e3) -- (v'1);

\node[circr,label=above:{\tiny $u$}] (u11) at (0,-1) {};
\node[circ,label=above:{\tiny $v$}] (v11) at (1.5,-1) {};

\draw[-Implies,line width=.6pt,double distance=2pt] (2.5,-1) -- (3.5,-1);

\node[circr,label=above:{\tiny $u$}] (u'11) at (4.5,-1) {};
\node[circ,right of=u'11,label=above:{\tiny $e_1$}] (e11) {}; 
\node[circ,right of=e11,label=above:{\tiny $e_2$}] (e21) {};
\node[circr,right of=e21,label=above:{\tiny $e_3$}] (e31) {};
\node[circ,right of=e31,label=above:{\tiny $v$}] (v'11) {};

\draw[-] (u11) -- (v11) node[above,midway] {\tiny $e$}
(u'11) -- (e11)
(e11) -- (e21)
(e21) -- (e31)
(e31) -- (v'11);

\node[circr,label=above:{\tiny $u$}] (u12) at (0,-2) {};
\node[circr,label=above:{\tiny $v$}] (v12) at (1.5,-2) {};

\draw[-Implies,line width=.6pt,double distance=2pt] (2.5,-2) -- (3.5,-2);

\node[circr,label=above:{\tiny $u$}] (u'12) at (4.5,-2) {};
\node[circr,right of=u'12,label=above:{\tiny $e_1$}] (e12) {}; 
\node[circ,right of=e12,label=above:{\tiny $e_2$}] (e22) {};
\node[circ,right of=e22,label=above:{\tiny $e_3$}] (e32) {};
\node[circr,right of=e32,label=above:{\tiny $v$}] (v'12) {};

\draw[-] (u12) -- (v12) node[above,midway] {\tiny $e$}
(u'12) -- (e12)
(e12) -- (e22)
(e22) -- (e32)
(e32) -- (v'12);
\end{tikzpicture}
\caption{Constructing a dominating set of $H$ from the dominating set $D$ of $G$ (vertices in red belong to the corresponding dominating set).}
\label{fig:dom3sub}
\end{figure}
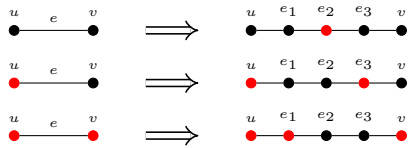

Conversely, let $D'$ be a minimum dominating set of $H$. First observe that at least one vertex among $e_1$, $e_2$ and $e_3$ belongs to $D'$ as $e_2$ must be dominated. Furthermore, we may assume, without loss of generality, that $\{e_1,e_3\} \not\subset D'$; indeed, if $\{e_1,e_3\} \subset D'$ then, by minimality of $D'$, $v \not\in D'$  for otherwise $D'\backslash \{e_3\}$ would be a dominating set of $G'$ of size strictly smaller than $D'$, a contradiction. But then, $(D'\backslash \{e_3\}) \cup \{v\}$ is a minimum dominating set of $G'$ not containing both $e_1$ and $e_3$. We next prove the following.

\begin{observation}
\label{obs:e1v}
If $e_1 \in D'$ (resp. $e_3 \in D'$) then $(D'\backslash \{e_1,e_2,e_3\}) \cup \{v\}$ (resp. $(D'\backslash \{e_1,e_2,e_3\}) \cup \{u\}$) is a dominating set of $G$ of size at most $\gamma (H) - 1$.
\end{observation}

Indeed, if $e_1 \in D'$ then either $v \in D'$ and $(D'\backslash \{e_1,e_2,e_3\}) \cup \{v\} = D'\backslash \{e_1,e_2,e_3\}$ is a dominating set of $G$ of size at most $\gamma (H) - 1$. Or $v \not\in D'$ but then $e_2\in D'$ since $e_3 \not\in D'$ (recall that $|D' \cap \{e_1,e_3\}| \leq 1$) must be dominated. But again, $(D'\backslash \{e_1,e_2,e_3\}) \cup \{v\}$ is a dominating set of $G$ of size at most $\gamma (H) - 1$. By symmetry, we conclude similarly if $e_3 \in D'$.

On the other hand, if $\{e_1,e_3\} \cap D' = \emptyset$, then $e_2 \in D'$ and $D'\backslash \{e_1,e_2,e_3\}$ is a dominating set of $G$ of size $\gamma (H) - 1$, which concludes the proof of the claim. \\

We next prove the statement of the lemma. Let $G'$ be the graph obtained from $G$ by 3-subdividing every edge of $G$. It then follows from Claim \ref{claim:3sub} that $\gamma (G') = \gamma (G) + |E|$. 

First assume that $G$ is a \yes-instance for \contracd{} i.e., there exists a minimum dominating set $D$ of $G$ containing an edge $e=uv$ (see Theorem \ref{theorem:contracdom}(i)). Let $D'$ be the minimum dominating set of $G'$ constructed according to the proof of Claim \ref{claim:3sub}. Then by construction, the edge $ue_1$ is contained in $D'$ which implies that $G'$ is a \yes-instance for \contracd{}.

Conversely, assume that $G'$ is a \yes-instance for \contracd{} that is, there exists a minimum dominating set $D'$ of $G'$ containing an edge $f$ (see Theorem \ref{theorem:contracdom}(i)). First note that we may assume that for any edge $e=uv\in E$, $\{e_1,e_3\} \not\subset D'$; indeed, if $\{e_1,e_3\} \subset D'$, then by minimality of $D'$ we have that $v \not\in D'$ (with $v$ adjacent to $e_3$) for otherwise $D' \backslash \{e_3\}$ is a dominating set of $G'$ of size strictly smaller than $D'$, a contradiction (also note that by minimality of $D'$, $e_2 \not\in D'$). But then, $(D'\backslash \{e_3\}) \cup \{v\}$ is also a minimum dominating set of $G'$ containing the edge $f$; indeed, since both $e_2$ and $v$ are not contained in $D'$, $e_3$ is not an endvertex of $f$. In the following, we denote by $e=uv$ the edge of $G$ such that $f$ is an edge of the 3-subdivision of $e$, with $e_1$ adjacent to $u$ and $e_3$ adjacent to $v$. 

Now consider the minimum dominating set $D$ of $G$ constructed according to the proof of Claim \ref{claim:3sub}. We distinguish two cases depending on whether $f=ue_1$ or $f=e_1e_2$ (note that the cases where $f=e_3v$ or $f=e_2e_3$ are symmetric to those considered).  

First assume that $f=ue_1$. Then, by Observation \ref{obs:e1v}, $v \in S$ and thus, $uv$ is an edge contained in $D$. Now, if $f=e_1e_2$ then again, by Observation \ref{obs:e1v}, $v\in D$. But then, by minimality of $D'$, we know that $e_3 \not\in D'$ as well as $v \not\in D'$, for otherwise $D'\backslash \{e_2\}$ would be a dominating set of $G'$ of size strictly smaller than $D'$, a contradiction. Thus, $v$ is dominated in $G'$ by some vertex $e'_1$ with $e'=vw\in E$, and it follows from Observation \ref{obs:e1v} that $w \in D$. But then, $D$ contains the edge $vw$, which concludes the proof.

\section{The proof of Theorem \ref{thm:nop}}
\label{section:nop}

We formally define the following problem.

\begin{center}
\begin{boxedminipage}{.99\textwidth}
\textsc{Edge Contraction($\pi$)}\\[2pt]
\begin{tabular}{ r p{0.8\textwidth}}
\textit{~~~~Instance:} &A graph $G=(V,E)$ and an edge $e \in E$.\\
\textit{Question:} &Is $\gamma(G\backslash e) \leq \gamma(G) -1$?
\end{tabular}
\end{boxedminipage}
\end{center}

We show that if \edgecontracd{} can be solved in polynomial time, then \dom{} can also be solved in polynomial time. Since \dom{} is a well-known $\mathsf{NP}$-complete problem, the result follows. 

Let $(G,\ell)$ be an instance for \dom{} and let $e$ be an edge of $G$. We run the polynomial time algorithm for \edgecontracd{} to determine if $\gamma(G\setminus e) =\gamma(G)-1$; we then have two possible scenarios.\\

\noindent{\bf Case 1.} $(G,e)$ is a \yes-instance for \edgecontracd{}. Since $\gamma(G\setminus e) =\gamma(G)-1$, we know that $G$ has a dominating set of size $\ell$ if and only if $G\setminus e$ has a dominating set of size $\ell -1$. Hence, we obtain that $(G\setminus e, \ell-1)$ is an equivalent instance for \dom{}. \\

\noindent{\bf Case 2.} $(G,e)$ is a \no-instance for \edgecontracd{}. Since $\gamma(G\setminus e) =\gamma(G)$, we know that $G$ has a dominating set of size $\ell$ if and only if $G\setminus e$ has a dominating set of size $\ell$. In this case, we obtain that $(G\setminus e, \ell)$ is an equivalent instance for \dom{}. \\

In both cases, the ensuing equivalent instance has one less vertex. Thus, by applying the polynomial-time algorithm for \edgecontracd{} at most $n$ times, we obtain a trivial instance for \dom{} and can therefore correctly determine its answer.

\section{The proof of Proposition \ref{prop:boundeddom}}
\label{section:easy}

\begin{itemize}
\item[(a)] It suffices to note that if we can compute $\gamma(G)$ and $\gamma(G\setminus e)$, for any edge $e$ of $G$, in polynomial time, then we can determine whether a graph $G$ is a \yes-instance for \contracd{} in polynomial time (we may proceed in a similar fashion for \twocontracd{}).

\item[(b)] Given a graph $G$ of $\mathcal{C}$, we first check whether $G$ has a dominating vertex. If it is the case, then $G$ is a \no-instance for \kcontracd{} for both $k=1,2$. Otherwise, we may consider any $S \subset V(G)$ with $\vert S \vert \leq q$ and check whether it is a dominating set of $G$. Since there are at most $\Oh(n^q)$ possible such subsets, we can determine the domination number of $G$ and check whether the conditions given in Theorem~\ref{theorem:contracdom} (i) or (ii) are satisfied in polynomial time. 

\item[(c)] The algorithm works similarly for $k=1$ and $k=2$. Let $H$ and $q$ be as stated and let $G$ be an instance of \kcontracd{} on $(H+K_1)$-free graphs. We first test whether $G$ is $H$-free (note that this can be done in time $\Oh(n^q)$). If this is the case, we use the polynomial-time algorithm for \kcontracd{} on $H$-free graphs. Otherwise, $G$ has an induced subgraph isomorphic to $H$; but since $G$ is a $(H+K_1)$-free graph, $V(H)$ must then be a dominating set of $G$ and so, $\gamma(G)\leq q$. We then conclude by Proposition~\ref{prop:boundeddom}(b) that \kcontracd{} is also polynomial-time solvable in this case.
\end{itemize}

\section{The proof of Observation \ref{obs:mimwcontrac}}
\label{section:mimw}
Indeed, note that the graph $G\setminus e$ can be obtained from $G$ by the removal of the vertices $u$ and $v$ where $e=uv$, and the addition of a new vertex whose neighborhood is $N(u)\cup N(v)$. The result then follows from Observation~\ref{obs:mimaddv} and the fact that vertex deletion does not increase the mim-width of a graph.


\begin{thebibliography}{99}
	
\bibitem{BTT11}
C. Bazgan, S. Toubaline, Z. Tuza, 
The most vital nodes with respect to stable set and vertex cover, 
Discrete Applied Mathematics 159 (2011) 1933--1946.

\bibitem{bazgan2013critical}
C. Bazgan, S. Toubaline, D. Vanderpooten, 
Critical edges for the assignment problem: Complexity and exact resolution, 
OR Letters 41(6) (2013) 685--689.

\bibitem{Bentz}
C. Bentz, M.-C. Costa, D. de Werra, C. Picouleau, B. Ries,
Weighted transversals and blockers for some optimization problems in graphs,
Progress in Combinatorial Optimization, Wiley-ISTE, 2012.

\bibitem{splitdom}
A.A. Bertossi,
Dominating sets for split and bipartite graphs,
Information Processing Letters 19(1) (1984) 37--40.

\bibitem{CWP11}
M.-C. Costa, D. de Werra, C. Picouleau, 
Minimum $d$-blockers and $d$-transversals in graphs, 
Journal of Combinatorial Optimization 22 (2011) 857--872.

\bibitem{CYGAN}
M. Cygan, F. V. Fomin, L. Kowalik, D. Lokshtanov, D. Marx, M. Pilipczuk, M. Pilipczuk, S. Saurabh
Parameterized Algorithms
Springer, 2015

\bibitem{dahl}
E. Dahlhaus, D. Johnson, C. Papadimitriou, P. Seymour, M. Yannakakis,
The Complexity of Multiterminal Cuts,
SIAM Journal on Computing 23(4) (1994) 864--894.

\bibitem{Di05}
R.~Diestel, 
Graph Theory, 
Springer-Verlag, 2005.

\bibitem{DPPR15}
O. Diner, D. Paulusma, C. Picouleau, B. Ries, 
Contraction blockers for graphs with forbidden induced paths, 
Proc. CIAC 2015,  LNCS 9079 (2015) 194--207.

\bibitem{diner2018contraction}
O. Diner, D. Paulusma, C. Picouleau, B. Ries, 
Contraction and Deletion Blockers for Perfect Graphs and H-free Graphs, 
Theoretical Computer Science 746 (2018) 49--72.

\bibitem{FOMIN18}
F. V. Fomin and P. A. Golovach, J.-F. Raymond, 
On the Tractability of Optimization Problems on H-Graphs, 
Proc. ESA 2018,  LIPIcs 112 (2018) 30:1--30:14.

\bibitem{garey}
M. R. Garey, D. S. Johnson,
Computers and Intractability; A Guide to the Theory of NP-Completeness,
W. H. Freeman \& Co.,New York, NY, USA (1990).

\bibitem{haynes}
T. W. Haynes and S. T. Hedetniemi, P. J. Slater,
Fundamentals of domination in graphs,
Pure and Applied Mathematics, Marcel Dekker, Vol. 208, 1998, New York, NY: Marcel Dekker, Inc.

\bibitem{HX10}
J. Huang, J.-M. Xu,
Domination and total domination contraction numbers of graphs,
Ars Combinatoria 94 (2010) 431--443.

\bibitem{PBP}
F.M. Pajouh, V. Boginski, E.L. Pasiliao, Minimum vertex blocker clique problem, Networks 64 (2014) 48--64.

\bibitem{PPR16}
D. Paulusma, C. Picouleau, B. Ries, 
Reducing the clique and chromatic number via edge contractions and vertex deletions, 
Proc. ISCO 2016, LNCS 9849 (2016) 38--49.

\bibitem{paulusma2017blocking}
D. Paulusma, C. Picouleau, B. Ries, 
Blocking stable sets for H-free graphs via edge contractions and vertex deletions, 
Proc. TAMC 2017, LNCS 9849 (2017) 470--483.
	
\bibitem{paulusma2018critical}
D. Paulusma, C. Picouleau, B. Ries, 
Critical vertices and edges in H-free graphs
Discrete Applied Mathematics, doi.org/10.1016/j.dam.2018.08.016.

\bibitem{RBPDCZ10}
B. Ries, C. Bentz, C. Picouleau, D. de Werra, M.-C. Costa, R. Zenklusen, 
Blockers and transversals in some subclasses of bipartite graphs : when caterpillars are dancing on a grid, 
Discrete Mathematics 310 (2010) 132--146.

\bibitem{VATSHELLE}
M. Vatshelle,
New Width Parameters in Graphs
PhD Thesis, University of Bergen, Norway, 2012.

\end{thebibliography}
\end{document}